\documentclass[11pt]{article}
\usepackage{MastenShortcuts}
\usepackage{color}

\usepackage{wrapfig}
\usepackage{tikz-cd}
\usepackage{tikz}
\usepackage{pgfplots}

\RequirePackage{amsthm}

\theoremstyle{definition}

\newtheorem{proposition}{Proposition}
\newtheorem{lemma}{Lemma}

\newtheorem{definition}{Definition}
\newtheorem{theorem}{Theorem}
\newtheorem{corollary}{Corollary}

\newenvironment{assump}[2][Assumption]{\begin{trivlist}
\item[\hskip \labelsep {\bfseries #1}\hskip \labelsep {\bfseries #2}]}{\end{trivlist}}

\newcounter{partialIndepSection}
\newcommand{\aAssump}{A\arabic{partialIndepSection}}

\newtheoremstyle{theoremSuppressedNumber}{}{}{}{}{\bfseries}{.}{ }{\thmname{#1}\thmnote{ (\mdseries #3)}}
\theoremstyle{theoremSuppressedNumber}
\newtheorem{partialIndepAssump}{Assumption \aAssump \addtocounter{partialIndepSection}{1}}

\usepackage{setspace}
\usepackage{natbib}
\usepackage{ulem}
\usepackage{subfig}
\usepackage{wrapfig}
\usepackage{float}
\usepackage{epstopdf}

\title{\textbf{Identification of Treatment Effects under \\ Conditional Partial Independence}\footnote{This paper is based on portions of our previous working papers, \cite{MastenPoirier2016,MastenPoirier2017}. We thank audiences at various seminars, as well as Federico Bugni, Ivan Canay, Joachim Freyberger, Joe Hotz, Guido Imbens, Shakeeb Khan, Chuck Manski, Jim Powell, Adam Rosen, Suyong Song, and Alex Torgovitsky, for helpful conversations and comments.}}

\author{Matthew A. Masten\footnote{Department of Economics, Duke University,
        \texttt{matt.masten@duke.edu}} \qquad Alexandre Poirier\thanks{
    Department of Economics, University of Iowa,
    \texttt{alexandre-poirier@uiowa.edu}}
}

\begin{document}
\maketitle
\begin{abstract}
Conditional independence of treatment assignment from potential outcomes is a commonly used but nonrefutable assumption. We derive identified sets for various treatment effect parameters under nonparametric deviations from this conditional independence assumption. These deviations are defined via a conditional treatment assignment probability, which makes it straightforward to interpret. Our results can be used to assess the robustness of empirical conclusions obtained under the baseline conditional independence assumption.
\end{abstract}

%
% C12 Hypothesis Testing: General
% C13 Estimation: General
% C14 Semiparametric and Nonparametric Methods
% C18 Methodological Issues: General
% C21 Cross-Sectional Models; Spatial Models;
% C25 Discrete Regression and Qualitative Choice Models; Discrete Regressors; Proportions Treatment Effect Models; Quantile Regressions
% C26 Instrumental Variables (IV) Estimation
% C51 Model Construction and Estimation


% \include puts it on a new page
% \input does not.
% So, when a section is finished, switch it from include to input

\onehalfspacing
\normalsize
\newpage
\section{Introduction}

The treatment effect model under conditional independence is widely used in empirical research. The conditional independence assumption states that, after conditioning on a set of observed covariates, treatment assignment is independent of potential outcomes. This assumption has many other names, including unconfoundedness, ignorability, exogenous selection, and selection on observables. It delivers point identification of many parameters of interest, including the average treatment effect, the average effect of treatment on the treated, and quantile treatment effects. \cite{ImbensRubin2015} provide a recent overview of this literature.

Without additional data, like the availability of an instrument, the conditional independence assumption is not refutable: The data alone cannot tell us whether the assumption is true. Moreover, conditional independence is often considered a strong and controversial assumption. Consequently, empirical researchers may wonder: How credible are treatment effect estimates obtained under conditional independence?

In this paper, we address this concern by studying what can be learned about treatment effects under a nonparametric class of assumptions that are weaker than conditional independence. While there are many ways to weaken independence, we focus on just one, which we call \emph{conditional $c$-dependence}.\footnote{See \cite{MastenPoirier2016} for an analysis and discussion of several other approaches. We refer to any of these approaches as \emph{partial} independence assumptions.} This assumption states that the probability of being treated given observed covariates and an unobserved potential outcome is not too far from the probability of being treated given just the observed covariates. We use the sup-norm distance, where the scalar $c$ denotes how much these two conditional probabilities may differ. This class of assumptions nests both the conditional independence assumption and the opposite end of no constraints on treatment selection.

In our first main contribution, we derive sharp bounds on conditional cdfs that are consistent with conditional $c$-dependence. This result can be used in many models, including the treatment effects model we study here.\footnote{See \cite{MastenPoirier2016} for several other applications of this result.} In that model, as our second main contribution, we derive identified sets for many parameters of interest. These include the average treatment effect, the average effect of treatment on the treated, and quantile treatment effects. These identified sets have simple, analytical characterizations. Empirical researchers can use these identified sets to examine how sensitive their parameter estimates are to deviations from the baseline assumption of conditional independence. We illustrate this sensitivity analysis in a brief numerical example.

\subsection*{Related literature}

In the rest of this section, we review the related literature. As discussed in section 22.4 of \cite{ImbensRubin2015}, a large literature starting with the seminal work of \cite{RosenbaumRubin1983sensitivity} relaxes conditional independence by modeling the conditional probabilities of treatment assignment given both observable and unobservable variables parametrically. This literature also typically imposes a parametric model on outcomes. This includes \cite{LinPsatyKronmal1998}, \cite{Imbens2003}, and \cite{AltonjiElderTaber2005,AltonjiElderTaber2008}. An important exception is \cite{RobinsRotnitzkyScharfstein2000}, who relax parametric assumptions on outcomes. They continue to use parametric models for treatment assignment probabilities, however, when applying their results. Our work builds on this literature by developing fully nonparametric methods for sensitivity analysis. Our new methods can ensure that empirical findings of robustness do not rely on auxiliary parametric assumptions.

We are aware of only two previous analyses in this sensitivity analysis literature which develop fully nonparametric methods. The first is \cite{IchinoMealliNannicini2008}, who avoid specifying a parametric model by assuming that all observed and unobserved variables are discretely distributed, so that their joint distribution is determined by a finite dimensional vector. Unlike our approach, theirs rules out continuous outcomes. It also involves many different sensitivity parameters, while our approach uses only one sensitivity parameter.

The second is \cite{Rosenbaum1995,Rosenbaum2002}, who proposes a sensitivity analysis within the context of randomization inference for testing the sharp null hypotheses of no unit level treatment effects for all units in one's dataset. Since this approach is based on finite sample randomization inference (c.f., chapter 5 of \citealt{ImbensRubin2015}), rather than population level identification analysis, this is quite different from the approaches discussed above and from what we do in the present paper. Like our results, however, his approach does not impose a parametric model on treatment assignment probabilities.

A large literature initiated by Manski has studied identification problems under various assumptions which typically do not point identify the parameters (e.g., \citealt{Manski2007}). In the context of missing data analysis, \cite{Manski2016} suggested imposing a class of assumptions which includes conditional $c$-dependence. He did not, however, derive any identified sets under this assumption. Several papers study partial identification of treatment response under deviations from mean independence assumptions, rather than the statistical independence assumption we start from. \cite{ManskiPepper2000,ManskiPepper2009} relax mean independence to a monotonicity constraint in the conditioning variable, while \cite{HotzMullinSanders1997} suppose mean independence only holds for some portion of the population. These relaxations and conditional $c$-dependence are non-nested. Moreover, these papers focus on mean potential outcomes, while we also study quantiles and distribution functions. Finally, Manski's original no assumptions bounds for average treatment effects (\citealt{Manski1989,Manski1990}) are obtained as a special case of our conditional $c$-dependence ATE bounds when $c$ is sufficiently large.

\section{Model, assumptions, and interpretation}\label{sec:Model}

We study the standard potential outcomes model with a binary treatment. In this section we setup the notation and some maintained assumptions. We define our parameters of interest and state the key assumption which point identifies them: random assignment of treatment, conditional on covariates. We discuss how we relax this assumption. We derive identified sets under these relaxations in section \ref{sec:identification}. We conclude this section by suggesting a few ways to interpret our deviations from conditional independence.

\subsection*{Basic setup}

Let $Y$ be an observed scalar outcome variable and $X \in \{ 0,1\}$ an observed binary treatment. Let $Y_1$ and $Y_0$ denote the unobserved potential outcomes. As usual the observed outcome is related to the potential outcomes via the equation
\begin{equation}\label{eq:potential outcomes}
	Y = X Y_1 + (1-X) Y_0.
\end{equation}

Let $W \in \supp(W)$ denote a vector of observed covariates, which may be discrete, continuous, or mixed. Let $p_{x \mid w} = \Prob(X=x \mid W=w)$ denote the observed generalized propensity score (\citealt{Imbens2000}). We consider both continuous and binary potential outcomes. We begin with the continuous outcome case, where we maintain the following assumption on the joint distribution of $(Y_1,Y_0,X,W)$.

\setcounter{partialIndepSection}{1}
\begin{partialIndepAssump}\label{assn:continuity}
For each $x,x' \in \{0,1\}$ and $w \in \supp(W)$:
\begin{enumerate}
\item \label{A1_1} $Y_x \mid X=x', W=w$ has a strictly increasing and continuous distribution function on $\supp(Y_x \mid X=x', W=w)$.

\item \label{A1_2} $\supp(Y_x \mid X=x',W=w) = \supp(Y_x \mid W=w) = [\underline{y}_x(w),\overline{y}_x(w)]$ where $-\infty \leq \underline{y}_x(w) < \overline{y}_x(w) \leq \infty$.

\item \label{A1_3} $p_{1 \mid w} \in (0,1)$ for all $w \in \supp(W)$.
\end{enumerate}
\end{partialIndepAssump}

%Via this assumption, we restrict attention to continuous potential outcomes. In particular, b
By equation \eqref{eq:potential outcomes},
\[
	F_{Y \mid X,W}(y \mid x,w) = \Prob(Y_x \leq y \mid X=x,W=w)
\]
and hence A\ref{assn:continuity}.\ref{A1_1} implies that the distribution function of $Y \mid X=x,W=w$ is also strictly increasing and continuous. By the law of iterated expectations, the marginal distributions of $Y$ and $Y_x$ have the same properties as well. We consider the binary outcome case on page \pageref{sec:binaryOutcomes}.

A\ref{assn:continuity}.\ref{A1_2} states that the conditional support of $Y_x$ given $X=x',W=w$ does not depend on $x'$, and that this support is a possibly infinite closed interval. The first equality is a `support independence' assumption, which is implied by the standard conditional independence assumption. Since $Y \mid X=x,W=w$ has the same distribution as $Y_x \mid X=x,W=w$, this implies that the support of $Y \mid X=x,W=w$ equals that of $Y_x \mid W=w$. Consequently, the endpoints $\underline{y}_x(w)$ and $\overline{y}_x(w)$ are point identified. A\ref{assn:continuity}.\ref{A1_3} is a standard overlap assumption.

Define the \emph{conditional rank} random variables $R_1 = F_{Y_1 \mid W}(Y_1 \mid W)$ and $R_0 = F_{Y_0 \mid W}(Y_0 \mid W)$. For any $w \in \supp(W)$, $R_1 \mid W=w$ and $R_0 \mid W=w$ are uniformly distributed on $[0,1]$, since $F_{Y_1 \mid W}(\cdot \mid w)$ and $F_{Y_0 \mid W}(\cdot \mid w)$ are strictly increasing. Moreover, by construction, both $R_1$ and $R_0$ are independent of $W$. The value of unit $i$'s conditional rank random variable $R_x$ tells us where unit $i$ lies in the conditional distribution of $Y_x \mid W=w$. We occasionally use these variables throughout the paper.

\subsection*{Identifying assumptions}

It is well known that the conditional distributions of potential outcomes $Y_1 \mid W$ and $Y_0 \mid W$ and therefore the marginal distributions of $Y_1$ and $Y_0$ are point identified under the following assumption:
\begin{itemize}
\item[] Conditional Independence: $X \independent Y_1 \mid W$ and $X \independent Y_0 \mid W$.
\end{itemize}
These marginal distributions are immediately point identified from
\[
	F_{Y_x \mid W}(y \mid w) = F_{Y \mid X,W}(y \mid x,w)
\]
and
\[
	F_{Y_x}(y) = \int_{\supp(W)} F_{Y \mid X,W}(y \mid x,w) \; dF_W(w).
\]
Consequently, any functional of $F_{Y_1 \mid W}$ and $F_{Y_0 \mid W}$ is also point identified under the conditional independence assumption. Leading examples include the average treatment effect,
\[
	\text{ATE} = \Exp(Y_1 - Y_0),
\]
and the $\tau$-th quantile treatment effect,
\[
	\text{QTE}(\tau) = Q_{Y_1}(\tau) - Q_{Y_0}(\tau),
\]
where $\tau \in (0,1)$. The goal of our identification analysis is to study what can be said about such functionals when conditional independence partially fails. To do this we define the following class of assumptions, which we call \emph{conditional $c$-dependence}.

\begin{definition}\label{def:c-dep}
Let $x \in \{ 0, 1 \}$. Let $w\in\supp(W)$. Let $c$ be a scalar between 0 and 1. Say $X$ is \emph{conditionally $c$-dependent} with $Y_x$ given $W=w$ if
\begin{equation}\label{eq:c-indep1}
	\sup_{y_x \in \supp(Y_x \mid W=w)} | \Prob(X=1 \mid Y_x=y_x,W=w) - \Prob(X=1 \mid W=w) | \leq c.
\end{equation}
If \eqref{eq:c-indep1} holds for all $w \in \supp(W)$ we say $X$ is conditionally $c$-dependent with $Y_x$ given $W$.
\end{definition}

Under the conditional independence assumption $X \independent Y_x \mid W$,
\[
	\P(X=1 \mid Y_x = y_x,W=w) = \Prob(X=1 \mid W=w)
\]
for all $y_x \in \supp(Y_x \mid W=w)$ and all $w \in \supp(W)$. Conditional $c$-dependence allows for deviations from this assumption by allowing the conditional probability $\Prob(X=1 \mid Y_x  = y_x, W=w)$ to be different from the propensity score $p_{1 \mid w}$, but not too different. This class of assumptions nests conditional independence as the special case where $c=0$. Moreover, when $c \geq \max \{ p_{1 \mid w}, p_{0 \mid w} \}$, from \eqref{eq:c-indep1} we see that conditional $c$-dependence imposes no constraints on $\Prob(X=1 \mid Y_x=y_x, W=w)$. Values of $c$ strictly between zero and $\max\{p_{1 \mid w},p_{0 \mid w}\}$ lead to intermediate cases. These intermediate cases can be thought of as a kind of limited selection on unobservables, since the value of one's unobserved potential outcome $Y_x$ is allowed to affect the probability of receiving treatment.

Beginning with \cite{RosenbaumRubin1983sensitivity}, many papers use parametric models for unobserved conditional probabilities similar to $\P(X=1 \mid Y_x = y_x,W=w)$ to model deviations from conditional independence. For example, see \cite{RobinsRotnitzkyScharfstein2000} and \cite{Imbens2003}. In contrast, conditional $c$-dependence is a nonparametric class of assumptions. Our results therefore ensure that empirical findings of robustness do not depend on any auxiliary parametric assumptions.

By invertibility of $F_{Y_x \mid W}(\cdot \mid w)$ for each $x \in \{0,1\}$ and $w \in \supp(W)$ (assumption A\ref{assn:continuity}.\ref{A1_1}), equation \eqref{eq:c-indep1} is equivalent to
\[
	\sup_{r_x \in [0,1]} | \Prob(X = 1 \mid R_x = r_x,W=w) - \Prob(X=1 \mid W=w) | \leq c.
	\tag{\ref{eq:c-indep1}$^\prime$}
\]
Using this result, we obtain the following characterization of conditional $c$-dependence.

\begin{proposition}\label{lemma:c-dep_KS_equivalence}
Suppose A\ref{assn:continuity}.\ref{A1_1} holds. Then $X$ is conditionally $c$-dependent with the potential outcome $Y_x$ given $W$ if and only if
\begin{equation}\label{eq:jointMinusMarginalProductEquivalenceLemma}
	\sup_{x' \in \{ 0,1 \} } \sup_{r \in [0,1]} | f_{X,R_x \mid W}(x',r \mid w) - p_{x' \mid w} f_{R_x \mid W}(r \mid w) | \leq c.
\end{equation}
\end{proposition}

Proposition \ref{lemma:c-dep_KS_equivalence} shows that conditional $c$-dependence is equivalent to a constraint on the sup-norm distance between the joint pdf of $(X,R_x) \mid W$ and the product of the marginal distributions of $X \mid W$ and $R_x \mid W$. Although we do not pursue this here, this alternative characterization also suggests how to extend this concept to continuous treatments. Finally, we note that another equivalent characterization of conditional $c$-dependence obtains by replacing $X=1$ with $X=0$ in equation \eqref{eq:c-indep1}.

Throughout the rest of the paper we impose conditional $c$-dependence between $X$ and the potential outcomes given covariates:

\begin{partialIndepAssump}\label{assn:cdep}
%$X$ is $c$-dependent with the potential outcomes $Y_x$ conditionally on $W=w$.
$X$ is conditionally $c$-dependent with $Y_1$ given $W$ and with $Y_0$ given $W$.
\end{partialIndepAssump}

\subsection*{Interpreting conditional $c$-dependence}

Interpreting the deviations from one's baseline assumption is an important part of any sensitivity analysis. In this subsection we give several suggestions for how to interpret our sensitivity parameter $c$ in practice. 

Our first suggestion, going back to the earliest sensitivity analysis of \cite{CornfieldEtAl1959}, and used more recently in \cite{Imbens2003}, \cite{AltonjiElderTaber2005,AltonjiElderTaber2008}, and \cite{Oster2016}, is to use the amount of selection on observables to calibrate our beliefs about the amount of selection on unobservables. To formalize this idea in our context, recall that conditional $c$-dependence is defined using a distance between two conditional treatment probabilities: the usual propensity score $\Prob(X=1 \mid W=w)$ and that same conditional probability, except also conditional on an unobserved potential outcome $Y_x$. Hence the question is: How much does adding this extra conditioning variable affect the conditional treatment probability?

In the data, $Y_x$ is unobserved, so we cannot answer this question directly. But we can examine the impact of adding additional observed covariates on conditional treatment probabilities (assuming $K \equiv \dim(W) \geq 1$). Specifically, suppose we partition our vector of covariates $W$ into $(W_{-k},W_k)$ where $W_k$ is the $k$th component of $W$ and $W_{-k}$ is a vector of the remaining $K-$1 components. Define
\[
	\overline{c}_k = \sup_{w_{-k}} \, \sup_{w_k} | \Prob(X=1 \mid W = (w_{-k},w_k) ) - \Prob(X=1 \mid W_{-k}=w_{-k}) |
\]
where we take suprema over $w_k \in \supp(W_k \mid W_{-k} = w_{-k})$ and $w_{-k} \in \supp(W_{-k})$. $\overline{c}_k$ tells us the largest amount that the observed conditional treatment probabilities with and without the variable $W_k$ can differ.\footnote{If $K=1$ then one can compare $\Prob(X=1 \mid W=w)$ with $\Prob(X=1)$.} Less formally, it is a measure of the marginal impact of including the $k$th variable on treatment assignment, given that we have already included the vector $W_{-k}$. Similarly to \cite{CornfieldEtAl1959} and the subsequent literature, the idea is that if adding an extra observed variable creates variation $\overline{c}_k$, then it might be reasonable to expect that adding the unobserved variable $Y_x$ to our conditioning set may also create variation $\overline{c}_k$. In practice, one can compute and examine $\overline{c}_k$ for each $k$.

Our next suggestion is a variation which incorporates information on the distribution of $W$. Define
\[
	p_{1 \mid W}(w_{-k},w_k) = \Prob(X=1 \mid W=(w_{-k},w_k) )
\]
and
\[
	p_{1 \mid W_{-k}}(w_{-k}) = \Prob(X=1 \mid W_{-k}=w_{-k}).
\]
Rather than examining the largest point in the support of the random variable
\[
	| p_{1 \mid W}(W_{-k},W_k) - p_{1 \mid W_{-k}}(W_{-k}) |
\]
we could also consider quantiles of this distribution, such as the 50th, 75th, or 90th percentiles. One could also plot the distribution of this random variable for each $k$. %This distribution could be used as a prior for $c$ in an empirical Bayes model averaging analysis, although we do not explore this here.

These suggestions are a kind of nonparametric version of the implicit partial $R^2$'s used by \cite{Imbens2003} in his parametric model, or of the logit coefficients used by \cite{RosenbaumRubin1983sensitivity}. The overall idea is the same: We are trying to measure the partial effect of adding an extra conditioning covariate on the conditional probability of treatment.

Our final suggestion reiterates a point made by Rosenbaum (\citeyear{Rosenbaum2002rejoinder}, section 7): Precise quantitative interpretations of sensitivity parameters like $c$ are not always necessary. We can perform qualitative comparisons of robustness across different studies and datasets by comparing the corresponding bound functions, as we do in our numerical illustration on page \pageref{sec:numericalIllustration}. \cite{Imbens2003} made a similar point, stating that ``not \ldots\ all evaluations are equally sensitive to departures from the exogeneity assumption'' (page 126). Such rankings of studies in terms of their robustness may help one aggregate findings across different studies. We leave a formal study of this kind of robustness-adjusted meta-analysis to future work.

\section{Identification under conditional $c$-dependence}\label{sec:identification}

In this section, we study identification of treatment effects under conditional $c$-dependence. To do so, we start by deriving bounds on cdfs under generic $c$-dependence. We then apply these results to obtain sharp bounds on various treatment effect functionals.

\subsection{Partial identification of cdfs}\label{section:theSetF}

In this subsection we consider the relationship between a generic scalar random variable $U$ and a binary variable $X \in \{ 0,1 \}$. We derive sharp bounds on the conditional cdf of $U$ given $X$ when (1) the marginal distributions of $U$ and $X$ are known and (2) $X$ is $c$-dependent with $U$, meaning that
\begin{equation}\label{eq:c-indep2}
	\sup_{u \in \supp(U)} | \Prob(X=1 \mid U=u) - \Prob(X=1) | \leq c.
\end{equation}
In the next subsection we will condition on $W$ and apply this general result with $U = R_x$, the conditional rank variable to obtain sharp bounds on various treatment effect parameters.

Let $F_{U \mid X}(u \mid x) = \Prob(U\leq u \mid X=x)$ denote the unknown conditional cdf of $U$ given $X=x$. Let $F_U(u) = \Prob(U \leq u)$ denote the known marginal cdf of $U$. Let $p_x = \Prob(X=x)$ denote the known marginal probability mass function of $X$.

Define
\[
	\overline{F}_{U \mid X}^c(u \mid x) = \min\left\{F_U(u) + \frac{c}{p_x}\min\{F_U(u),1-F_U(u)\}, \, \frac{F_U(u)}{p_x}, \, 1\right\}
\]
and
\[
	\underline{F}_{U \mid X}^c(u \mid x) = \max\left\{F_U(u) - \frac{c}{p_x} \min\{F_U(u),1-F_U(u)\}, \, \frac{F_U(u)-1}{p_x} +1, \, 0 \right\}.
\]

\begin{theorem}\label{lemma:c-dep_cdf_bounds}
Suppose the following hold.
\begin{enumerate}
\item The marginal distributions of $U$ and $X$ are known. 

\item $U$ is continuously distributed.\label{thm1_A2}

\item Equation \eqref{eq:c-indep2} holds. 

\item $p_1 \in (0,1)$. 
\end{enumerate}
Let $\mathcal{F}_{\supp(U)}$ denote the set of all cdfs on $\supp(U)$. Then, for each $x \in \{ 0,1 \}$, $F_{U \mid X}(\cdot \mid x) \in \mathcal{F}_{U \mid x}^c$ where
\[ 
	\mathcal{F}_{U \mid x}^c =
	\left \{ F \in \mathcal{F}_{\supp(U)} : \underline{F}_{U \mid X}^c(u \mid x) \leq F(u) \leq \overline{F}_{U \mid X}^c(u \mid x) \text{ for all } u \in \supp(U) \right\}.
\]
Furthermore, for each $\epsilon \in [0,1]$ there exists a joint distribution of $(U,X)$ consistent with assumptions 1--4 above and such that
\begin{multline}\label{eq:intermediate R cdf bounds}
	\big( \Prob(U\leq u \mid X=1), \, \Prob(U\leq u \mid X=0) \big) \\
	= \left( \epsilon  \underline{F}_{U \mid X}^c(u \mid 1) + (1-\epsilon) \overline{F}_{U \mid X}^c(u \mid 1), \, 
		(1-\epsilon)  \underline{F}_{U \mid X}^c(u \mid 0) + \epsilon \overline{F}_{U \mid X}^c(u \mid 0) \right)
\end{multline}
for all $u \in \supp(U)$. Consequently, for any $x \in \{ 0,1 \}$ and $u \in \supp(U)$ the pointwise bounds
\[
	F_{U \mid X}(u \mid x) \in [\underline{F}_{U \mid X}^c(u \mid x), \overline{F}_{U \mid X}^c(u \mid x)]
\]
are sharp.
\end{theorem}

The proof of theorem \ref{lemma:c-dep_cdf_bounds}, along with all other proofs, is given in the appendix. In this proof, we note that there are two constraints on the conditional distribution of $U \mid X$. The first is the $c$-dependence constraint. The second is the fact that the marginal distributions of $U$ and $X$ are known, and hence the conditional cdfs must satisfy a law of total probability constraint. This result is therefore a variation on the decomposition of mixtures problem. See \cite{CrossManski2002}, \cite{Manski2007} chapter 5, and \cite{MolinariPeski2006} for further discussion.

Theorem 1 has three conclusions. First, we show that the functions $\overline{F}_{U \mid X}^c(\cdot \mid x)$ and $\underline{F}_{U \mid X}^c(\cdot \mid x)$ bound the unknown conditional cdf $F_{U \mid X}(\cdot \mid x)$ uniformly in their arguments. Second, we show that these bounds are functionally sharp in the sense that the joint identified set for the two conditional cdfs $(F_{U \mid X}(\cdot \mid 1),F_{U \mid X}(\cdot \mid 0))$ contains linear combinations of the bound functions $\overline{F}_{U \mid X}^c(\cdot \mid x)$ and $\underline{F}_{U \mid X}^c(\cdot \mid x)$. Finally, we remark that this functional sharpness implies pointwise sharpness.

Importantly, the bound functions $\overline{F}_{U \mid X}^c(\cdot \mid x)$ and $\underline{F}_{U \mid X}^c(\cdot \mid x)$ are piecewise linear functions with simple analytical expressions. These bounds are proper cdfs and can be attained, as stated above. As $c$ approaches zero, these bounds for $F_{U \mid X}(u \mid x)$ collapse to the conditional cdf $F_U(u)$. When $c$ exceeds $\max\{p_{0},p_{1}\}$, the $c$-dependence constraint is not binding. Consequently, the cdf bounds simplify to
\[
	\overline{F}_{U \mid X}^c(u \mid x) = \min\left\{\frac{F_U(u)}{p_{x}}, \, 1 \right\}
	\qquad \text{and} \qquad
	\underline{F}_{U \mid X}^c(u \mid x) = \max\left\{\frac{F_U(u)-1}{p_{x}} + 1, \, 0\right\}.
\]
These bounds can be interpreted as the \emph{no assumptions bounds} since the only constraint imposed on the cdfs is that they satisfy the law of total probability. 

Figure \ref{cIndepBoundsOnCdfs} shows several examples of the bound functions $\overline{F}_{U \mid X}^c(\cdot \mid x)$ and $\underline{F}_{U \mid X}^c(\cdot \mid x)$. In this example we let $U \sim \text{Unif}[0,1]$ and $p_1 = 0.75$. We let $c= 0.1$, $0.4$, and $0.9$, which represent three qualitative regions for $c$: (a) $c < \min \{ p_1, p_0 \}$, where both bounds are strictly between 0 and 1 on the interior of the support, (b) $\min \{ p_1, p_0 \} \leq c < \max \{ p_1, p_0 \}$, where one bound is strictly between 0 and 1 on the interior of the support, but the other is not, and (c) $c \geq \max \{ p_1, p_0 \}$, where we simply obtain the no assumptions bounds.

% , along with the functions $p_1(u) = \Prob(X=1 \mid U=u)$ corresponding to these bounds

\begin{figure}[t]
\centering
\includegraphics[width=50mm]{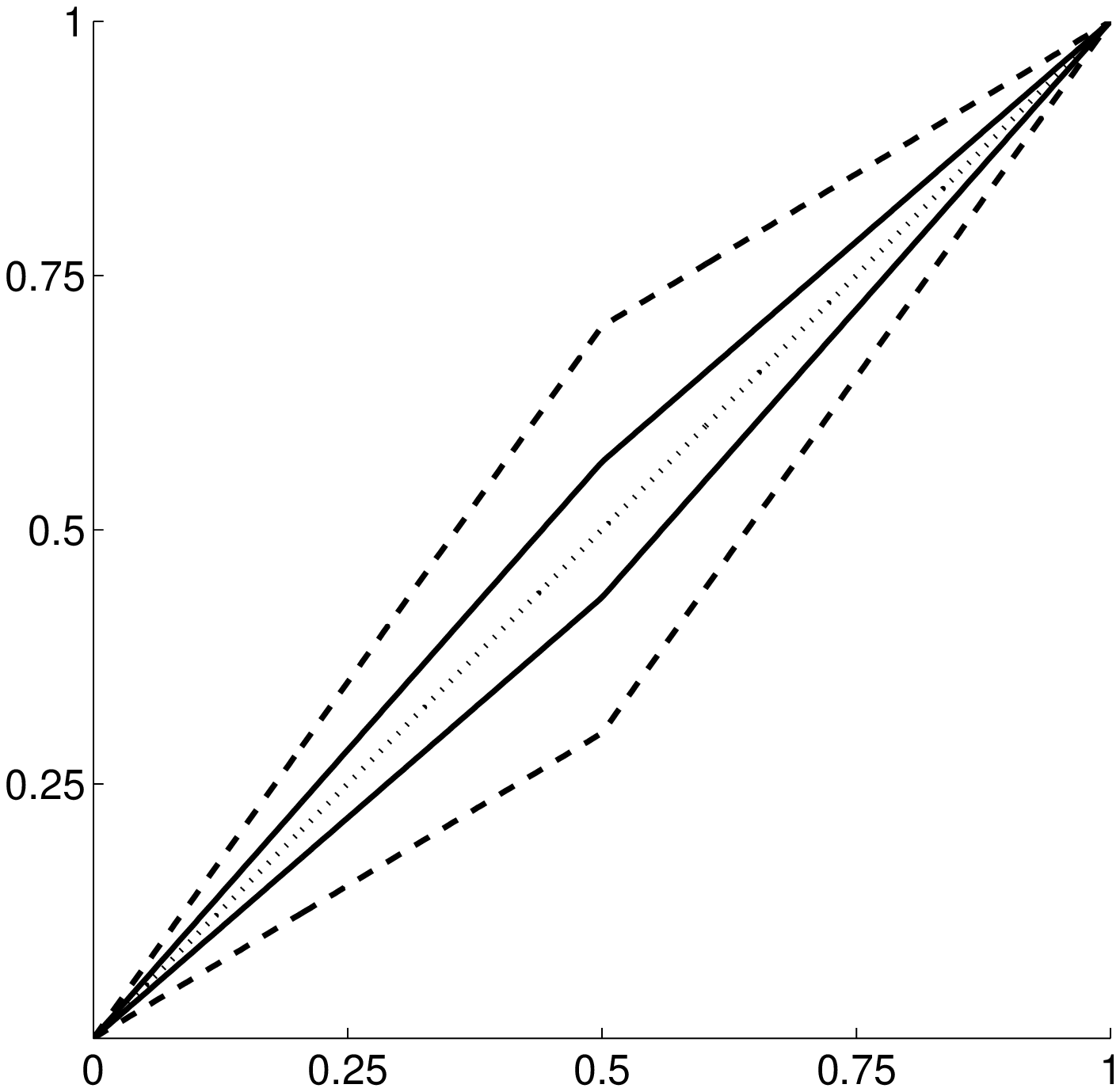}
\hspace{3mm}
\includegraphics[width=50mm]{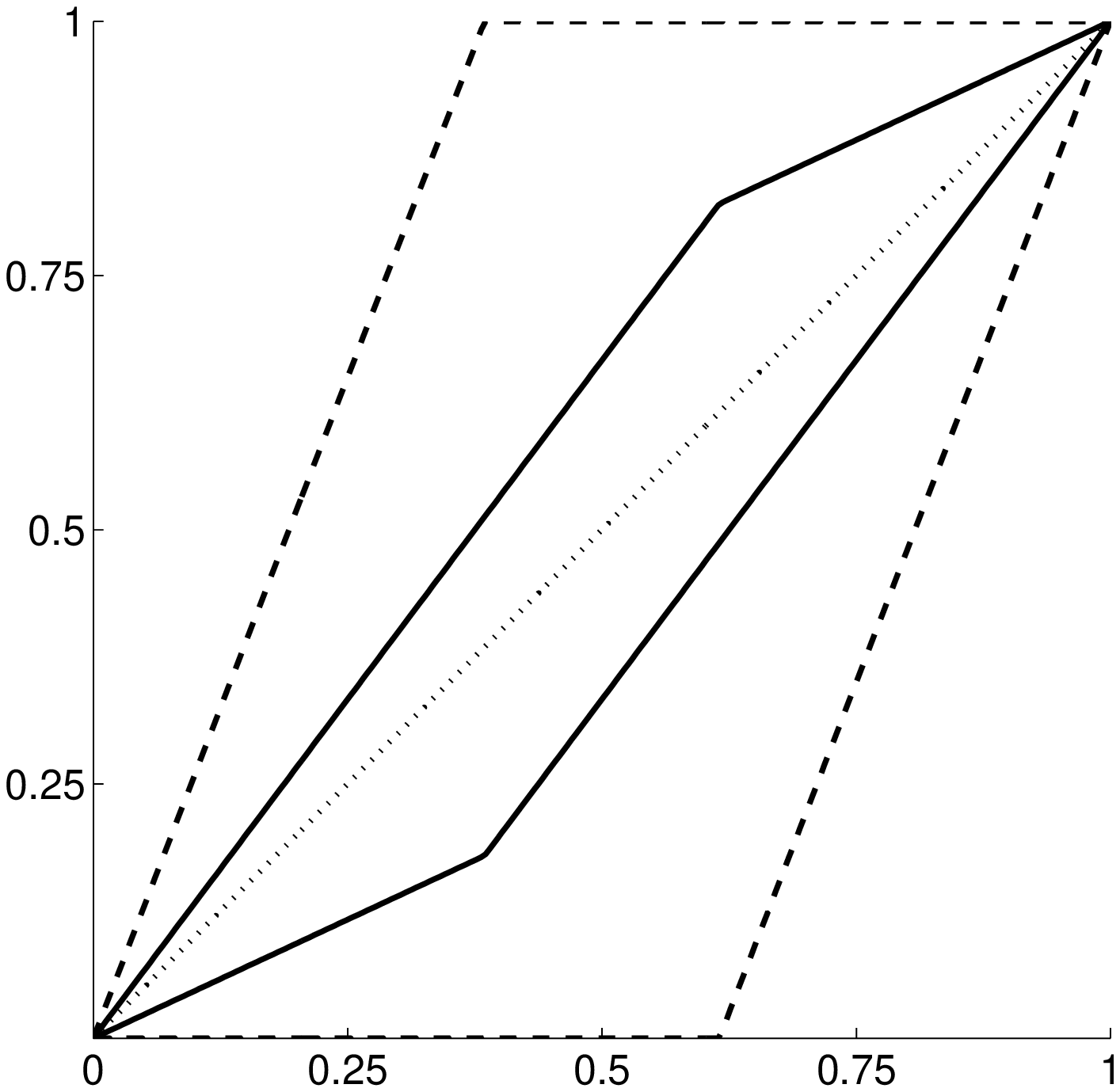}
\hspace{3mm}
\includegraphics[width=50mm]{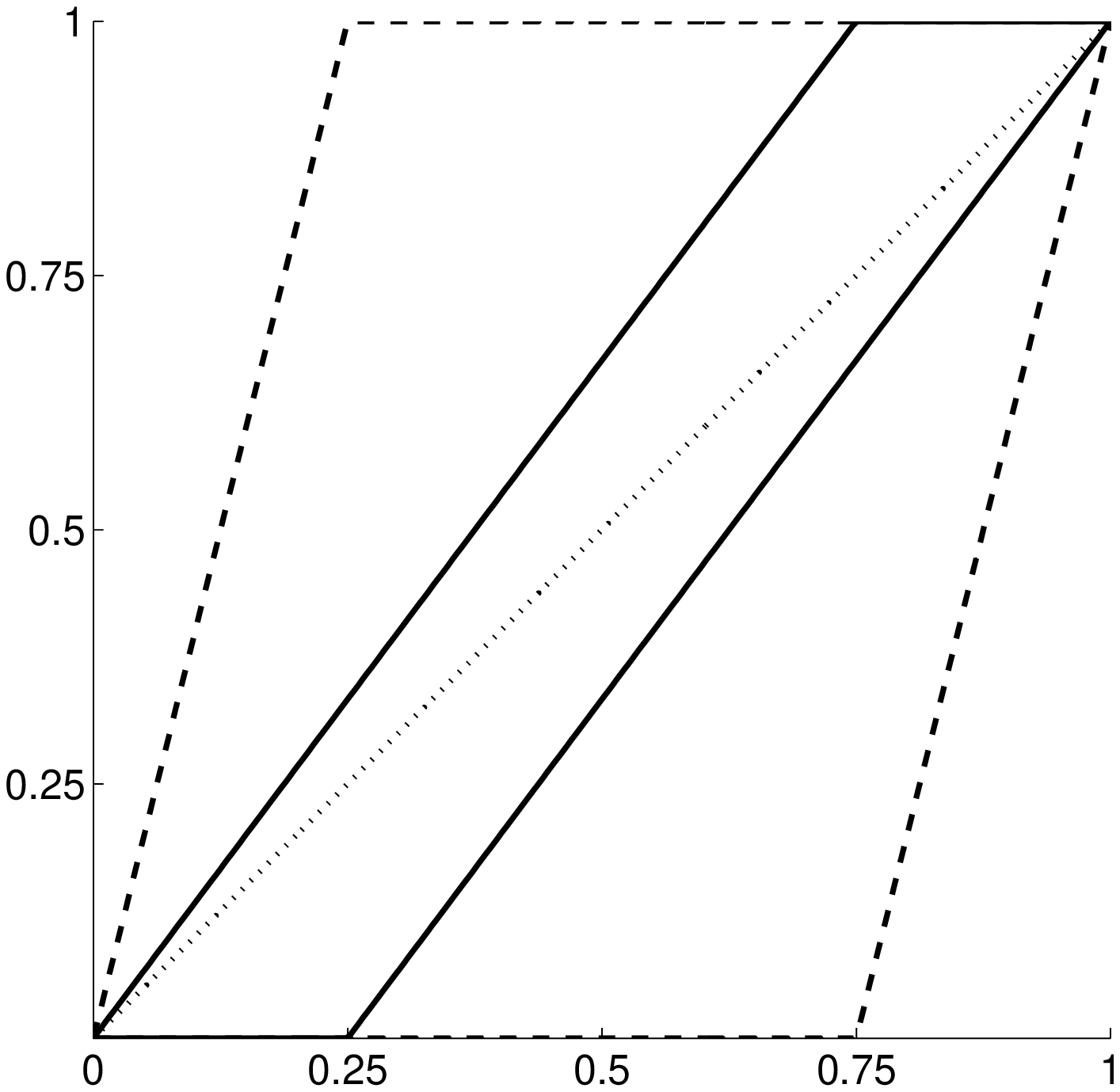}
\caption{Example upper and lower bounds on $F_{U \mid X}(u \mid x)$ for $x=1$ (solid) and $x=0$ (dashed), when $p_1 = 0.75$. Left: $c = 0.1 < \min \{ p_1, p_0 \}$. Middle: $\min \{ p_1, p_0 \} < c = 0.4 < \max \{ p_1, p_0 \}$. Right: $c = 0.9 > \max \{ p_1, p_0 \}$. The diagonal, representing the $c=0$ case of full independence, is plotted as a dotted line.}
\label{cIndepBoundsOnCdfs}
\end{figure}

\subsection{Partial identification of treatment effects}

Next we study identification of various treatment effects under conditional $c$-dependence. Throughout most of this section we focus on continuous outcomes. We study binary outcomes on page \pageref{sec:binaryOutcomes}. We end with a numerical illustration of the identified set for ATE as a function of $c$, and discuss how this set depends on features of the observed distribution of the data.

\subsubsection*{Conditional cdfs}

Under conditional independence, $c=0$, the marginal conditional distribution functions $F_{Y_0 \mid W}$ and $F_{Y_1 \mid W}$ are point identified. For $c > 0$, these functions are partially identified. In this case, we derive sharp bounds on these cdfs in proposition \ref{prop:condcdf_id} below.

Define
\begin{align}\label{eq:cdf upperbound}
	\overline{F}^c_{Y_x \mid W}(y \mid w)
	&= \min\left\{\frac{p_{x \mid w}F_{Y \mid X,W}(y \mid x,w)}{p_{x \mid w}-c}\indicator(p_{x \mid w} > c)+\indicator(p_{x \mid w} \leq c), \right.\\
	&\hspace{20mm} \left. \frac{p_{x \mid w}F_{Y \mid X,W}(y \mid x,w) + c}{p_{x \mid w} + c}, \, \, p_{x \mid w}F_{Y \mid X,W}(y \mid x,w) + (1-p_{x \mid w})\right\}\notag
\end{align}
and
\begin{align}\label{eq:cdf lowerbound}
	\underline{F}^c_{Y_x \mid W}(y \mid w)
	&=
	\max\left\{\frac{p_{x \mid w}F_{Y \mid X,W}(y \mid x,w)}{p_{x \mid w}+c}, \right.\\
	&\hspace{20mm} \left. \frac{p_{x \mid w}F_{Y \mid X,W}(y \mid x,w) - c}{p_{x \mid w} - c}\indicator(p_{x \mid w} > c), \, \, p_{x \mid w}F_{Y \mid X,W}(y \mid x,w)\right\}\notag
\end{align}
for all $y\in [\underline{y}_x(w),\overline{y}_x(w))$. For $y < \underline{y}_x(w)$, define these cdf bounds to be zero. For $y \geq \overline{y}_x(w)$, define these cdf bounds to be one. 

The following proposition shows that the functions \eqref{eq:cdf upperbound} and \eqref{eq:cdf lowerbound} are sharp bounds on the cdf of $Y_x \mid W=w$ under conditional $c$-dependence.

\begin{proposition}[Bounds on conditional cdfs]\label{prop:condcdf_id}
Let $w\in\supp(W)$. Suppose the joint distribution of $(Y,X,W)$ is known. Let A\ref{assn:continuity} and A\ref{assn:cdep} hold. Let $\mathcal{F}_\R$ denote the set of all cdfs on $\R$.
\begin{enumerate}
\item Then $F_{Y_x \mid W}(\cdot \mid w) \in \mathcal{F}_{Y_x \mid w}^c$ where
\[ 
	\mathcal{F}_{Y_x \mid w}^c =
	\left \{ F \in \mathcal{F}_\R : \underline{F}_{Y_x \mid W}^c(y \mid w) \leq F(y) \leq \overline{F}_{Y_x \mid W}^c(y \mid w)\text{ for all } y \in \R \right\}.
\]

\item Furthermore, for each $\epsilon \in [0,1]$ and $0 < \eta < \min \{ p_{x \mid w}, 1 - p_{x \mid w} \}$, there exist cdfs $F_{Y_x \mid W}^c(\cdot \mid w; \epsilon,\eta) \in \mathcal{F}_{Y_x \mid w}^c$ for $x \in \{0,1\}$ such that
\begin{enumerate}

\item There exists a joint distribution of $(Y_1,Y_0,X) \mid W=w$ consistent with our maintained assumptions and such that
\[
	\Prob(Y_1 \leq y \mid W=w) = F_{Y_1 \mid W}^c(y \mid w; \epsilon,\eta)
	\quad \text{and} \quad
	\Prob(Y_0 \leq y \mid W=w) = F_{Y_0 \mid W}^c(y \mid w; 1-\epsilon,\eta)
\]
for all $y \in \supp(Y_x \mid W=w)$.

\item For each $x \in \{0,1\}$, as $\eta \searrow 0$, $F_{Y_x \mid W}^c(\cdot \mid w;1,\eta)$ and $F_{Y_x \mid W}^c(\cdot \mid w; 0,\eta)$ converge pointwise monotonically to $\overline{F}_{Y_x \mid W}^c(\cdot \mid w)$ and $\underline{F}_{Y_x \mid W}^c(\cdot \mid w)$, respectively.

\item For each $y \in \R$ and $w \in \supp(W)$, the function $F_{Y_x \mid W}^c(y \mid w; \cdot,\eta)$ is continuous on $[0,1]$. 
%\begin{enumerate}
%\item 
%\item For $\epsilon=0$ or 1, the function $F_{Y_x \mid W}^c(y \mid w; \epsilon,\cdot)$ is continuous on $(0,\min \{ p_{x \mid w}, 1 - p_{x \mid w} \})$.
%\end{enumerate}
\end{enumerate}

\item Consequently, for any $y \in \R$, the pointwise bounds
\begin{align}\label{eq:condcdfinterval}
F_{Y_x \mid W}(y \mid w) &\in [\underline{F}_{Y_x \mid W}^c(y \mid w),\overline{F}_{Y_x \mid W}^c(y \mid w)]
\end{align}
have a sharp interior.
\end{enumerate}
\end{proposition}

The proof of this result relies importantly on our general result, theorem \ref{lemma:c-dep_cdf_bounds}. Similarly to that result, proposition \ref{prop:condcdf_id} has three conclusions. First we show that the functions \eqref{eq:cdf upperbound} and \eqref{eq:cdf lowerbound} bound $F_{Y_x \mid W}(\cdot \mid w)$ uniformly in their arguments. Second, we show that these bounds are functionally sharp. As in theorem \ref{lemma:c-dep_cdf_bounds}, sharpness is subtle because $\mathcal{F}_{Y_x \mid w}^c$ is not the identified set for $F_{Y_x \mid W}(\cdot \mid w)$---it contains some cdfs which cannot be attained. For example, it contains cdfs with jump discontinuities on their support, which violates A\ref{assn:continuity}.\ref{A1_1}. We could impose extra constraints to obtain the sharp set of cdfs but this is not required for our analysis. 

That said, the bound functions \eqref{eq:cdf upperbound} and \eqref{eq:cdf lowerbound} used to define $\mathcal{F}_{Y_x \mid W}^c$ \emph{are} sharp for the function $F_{Y_x \mid W}(\cdot \mid w)$ in the sense that there are cdfs $F_{Y_x \mid W}(\cdot \mid w; \epsilon, \eta)$ which are (a) attainable, (b) can be made arbitrarily close to the bound functions, and (c) continuously vary between the lower and upper bounds. The bound functions themselves are not always continuous and so violate A\ref{assn:continuity}.\ref{A1_1}. This explains the presence of the $\eta$ variable. It also explains why the endpoints of the bounds in our results below may not be attainable. If $c$ is small enough, the bound functions can be attainable, but we do not enumerate these cases for simplicity.

We use attainability of the functions $F_{Y_x \mid W}(\cdot \mid w; \epsilon, \eta)$ to prove sharpness of identified sets for various functionals of $F_{Y_1 \mid W}$ and $F_{Y_0 \mid W}$. For example, as the third conclusion to proposition \ref{prop:condcdf_id}, we already stated the pointwise-in-$y$ sharpness result for the evaluation functional. In general, we obtain bounds for functionals by evaluating the functional at the bounds \eqref{eq:cdf upperbound} and \eqref{eq:cdf lowerbound}. Sharpness of these bounds then follows by applying proposition \ref{prop:condcdf_id}.

\subsubsection*{Conditional QTEs, CATE, and ATE}

Next we derive identified sets for functionals of the marginal distribution of potential outcomes given covariates. We begin with the conditional quantile treatment effect:
\[
	\text{CQTE}(\tau \mid w) = Q_{Y_1 \mid W}(\tau \mid w) - Q_{Y_0 \mid W}(\tau \mid w).
\]
By integrating these bounds over $\tau$ from zero to one, we will obtain sharp bounds for the conditional average treatment effect: 
\[
	\text{CATE}(w) = \Exp(Y_1 \mid W=w) - \Exp(Y_0 \mid W=w).
\]
Finally, averaging these bounds over the marginal distribution of $W$ yields sharp bounds on ATE.
%	\text{ATE} = \Exp[\text{CATE}(W)] = \Exp(Y_1 - Y_0).

We first give closed form expressions for bounds on the quantile function of the potential outcome $Y_x$ given $W=w$. Define
\begin{equation}\label{eq:quantile upperbound}
	\overline{Q}^c_{Y_x \mid W}(\tau \mid w)
	= Q_{Y \mid X,W}\left(\min\left\{\tau + \frac{c}{p_{x \mid w}}\min\{\tau,1-\tau\},\frac{\tau}{p_{x \mid w}},1\right\}\mid x,w\right)
\end{equation}
and 
\begin{equation}\label{eq:quantile lowerbound}
	\underline{Q}^c_{Y_x \mid W}(\tau \mid w)
	=
	Q_{Y \mid X,W}\left(\max\left\{\tau - \frac{c}{p_{x \mid w}}\min\{\tau,1-\tau\},\frac{\tau-1}{p_{x \mid w}} +1,0\right\}\mid x,w\right).
\end{equation}
The following proposition and corollary formalize these results.

\begin{proposition}[Bounds on CQTE]\label{prop:CQTE_bounds}
Let $w \in \supp(W)$. Let A\ref{assn:continuity} and A\ref{assn:cdep} hold. Suppose the joint distribution of $(Y,X,W)$ is known. Let $\tau \in (0,1)$. Then $\text{CQTE}(\tau \mid w)$ lies in the set
\begin{align}\label{eq:CQTE_bounds}
	&\left[\underline{\text{CQTE}}^c(\tau \mid w), \overline{\text{CQTE}}^c(\tau \mid w)\right] \notag \\
	&\hspace{25mm} \equiv 
	\left[\underline{Q}^c_{Y_1 \mid W}(\tau \mid w) - \overline{Q}^c_{Y_0 \mid W}(\tau \mid w), \,\overline{Q}^c_{Y_1 \mid W}(\tau \mid w) - \underline{Q}^c_{Y_0 \mid W}(\tau \mid w)\right].
\end{align}
Moreover, the interior of this set is sharp.
\end{proposition}

The bounds \eqref{eq:CQTE_bounds} are also sharp for the function $\text{CQTE}(\cdot \mid \cdot)$ in a sense similar to that used in theorem \ref{lemma:c-dep_cdf_bounds} and proposition \ref{prop:condcdf_id}; we omit the formal statement for brevity. This functional sharpness delivers the following result.

\begin{corollary}[Bounds on CATE and ATE]\label{corollary:CATE_ATE_bounds}
Suppose the assumptions of proposition \ref{prop:CQTE_bounds} hold. Suppose $\Exp ( |Y| \mid X=x,W =w ) < \infty$ for all $(x,w) \in \supp(X,W)$.
\begin{enumerate}
\item Then $\text{CATE}(w)$ lies in the set
\[
	\left[\underline{\text{CATE}}^c(w),\overline{\text{CATE}}^c(w) \right]
	\equiv
	\left[\int_0^1 \underline{\text{CQTE}}^c(\tau \mid w) \; d\tau, \int_0^1 \overline{\text{CQTE}}^c(\tau \mid w) \; d\tau \right].
\]

\item Suppose further that $\Exp[ \Exp(|Y| \mid X=x,W) ] < \infty$ for $x \in \{0,1 \}$. Then $\text{ATE}$ lies in the set
\[
	\left[\underline{\text{ATE}}^c,\overline{\text{ATE}}^c \right]
	\equiv \Big[ \Exp \big( \underline{\text{CATE}}^c(W) \big), \, \Exp \big( \overline{\text{CATE}}^c(W) \big) \Big]
\]
assuming these means exist (including possibly $\pm \infty$).
\end{enumerate}
Moreover, the interiors of these sets are sharp.
\end{corollary}

All of these bounds are defined directly from equations \eqref{eq:quantile upperbound} and \eqref{eq:quantile lowerbound}, or averages of those equations. Those equations have simple analytical expressions, which makes all of these bounds quite tractable. These bounds are all monotonic in $c$, as illustrated in figure \ref{numericalIllustration_identifiedSets} of our numerical example. In particular, as $c$ goes to zero, the CQTE bounds collapse to the point $Q_{Y \mid X,W}(\tau \mid 1,w) - Q_{Y \mid X,W}(\tau \mid 0,w)$ while the $\text{CATE}(w)$ bounds collapse to the point $\Exp(Y \mid X=1, W=w) - \Exp(Y \mid X=0,W=w)$ and the ATE bounds collapse to $\Exp[ \Exp(Y \mid X=1, W) - \Exp(Y \mid X=0,W) ]$.

\subsubsection*{Unconditional cdfs and the unconditional QTE}

We can also derive bounds on the unconditional $\text{QTE}(\tau)$ by first deriving bounds on the unconditional cdfs of $Y_1$ and $Y_0$ and inverting them. These unconditional cdf bounds obtain by integrating the conditional bounds of proposition \ref{prop:condcdf_id} over $w$. Inverses here denote the left inverse.

\begin{corollary}[Bounds on marginal cdfs, quantiles, and QTEs]\label{lemma:margCDF_QTE_bounds}
Let A\ref{assn:continuity} and A\ref{assn:cdep} hold. Suppose the joint distribution of $(Y,X,W)$ is known. Let $y\in\R$ and $\tau \in (0,1)$. Then the following hold.
\begin{enumerate}
\item $F_{Y_x}(y)$ lies in the set
%\[
%	\left[\underline{F}^c_{Y_x}(y),\overline{F}^c_{Y_x}(y)\right]
%	\equiv
%	\left[\int_{\supp(W)}\underline{F}_{Y_x \mid W}^c(y \mid w) \; dF_W(w), \, \int_{\supp(W)}\overline{F}_{Y_x \mid W}^c(y \mid w) \; dF_W(w)\right]
%\]
\[
	\left[\underline{F}^c_{Y_x}(y),\overline{F}^c_{Y_x}(y)\right]
	\equiv
	\left[ \Exp \left( \underline{F}_{Y_x \mid W}^c(y \mid W) \right), \, \Exp \left( \overline{F}_{Y_x \mid W}^c(y \mid W) \right) \right].
\]

\item $Q_{Y_x}(\tau)$ lies in the set
\[
	\left[\underline{Q}^c_{Y_x}(\tau),\overline{Q}^c_{Y_x}(\tau)\right]
	\equiv
	\left[(\overline{F}^c_{Y_x})^{-1}(\tau), \, (\underline{F}_{Y_x}^c)^{-1}(\tau)\right].
\]

\item $\text{QTE}(\tau)$ lies in the set
\[
	\left[\underline{\text{QTE}}^c(\tau),\overline{\text{QTE}}^c(\tau)\right]
	\equiv
	\left[\underline{Q}_{Y_1}^c(\tau) - \overline{Q}_{Y_0}^c(\tau), \, \overline{Q}_{Y_1}^c(\tau) - \underline{Q}_{Y_0}^c(\tau)\right].
\]
\end{enumerate}
Moreover, the interiors of these sets are sharp.
\end{corollary}

As with our earlier results, all three results in this corollary are also functionally sharp. And again all these bounds collapse to a single point as $c$ approaches zero.

\subsection*{The ATT}\label{sec:ATT}

The average effect of treatment on the treated is
\[
	\text{ATT} = \Exp(Y_1 - Y_0 \mid X=1).
\]
Under conditional independence, $\text{ATT} = \Exp[ \text{CATE}(W) \mid X=1 ]$. That is, we average CATE over the distribution of covariates $W$ within the treated group, whereas ATE is the unconditional average of CATE over the covariates, $\text{ATE} = \Exp[ \text{CATE}(W) ]$. In corollary \ref{corollary:CATE_ATE_bounds} we showed that the bounds on ATE under conditional $c$-dependence are simply the average of our CATE bounds over the marginal distribution of $W$. Hence a natural first approach for obtaining bounds on ATT is to average our CATE bounds over the distribution of $W \mid X=1$, just as we do in the baseline case of $c=0$. For $c > 0$, however, this approach is not correct. This follows since, when conditional independence fails, potential outcomes are not independent of treatment assignment, even conditional on covariates. That is, for $x \in \{0,1\}$, the distribution of $Y_x \mid X=1,W=w$ is not necessarily the same as the distribution of $Y_x \mid X=0,W=w$. Hence the parameters $\Exp(Y_1-Y_0 \mid W=w)$ and $\Exp(Y_1 - Y_0 \mid X=1,W=w)$ are not necessarily equal. Below, we derive the correct identified set for ATT under conditional $c$-dependence.

%This is because $\text{ATT} = \Exp[\Exp(Y_1 - Y_0 \mid X=1,W) \mid X=1]$, and the term $\Exp(Y_1 - Y_0 \mid X=1,W=w)$ will not be equal to  $\text{CATE}(w)$ unless selection on observables holds.

As in the baseline case of conditional independence, equation \eqref{eq:potential outcomes} immediately implies that $\Exp(Y_1 \mid X=1)$ is point identified by $\Exp(Y \mid X=1)$ without any assumptions on the dependence between $Y_1$ and $X$. Hence only $\Exp(Y_0 \mid X=1)$ will be partially identified under deviations from conditional independence. Thus we relax A\ref{assn:continuity}.\ref{A1_3} and A\ref{assn:cdep} as follows.

\begin{assump}{A\ref{assn:continuity}.\ref{A1_3}$^\prime$}
$p_1 > 0$ and, for all $w \in \supp(W)$, $p_{1 \mid w} < 1$.
%
%$p_{1 \mid w} < 1$ for all $w \in \supp(W)$.
\end{assump}

\begin{assump}{A\ref{assn:cdep}$^\prime$}
$X$ is conditionally $c$-dependent with $Y_0$ given $W$.
\end{assump}

By the law of iterated expectations and some algebra,
\[
	\Exp(Y_0 \mid X=1) = \frac{\Exp(Y_0) - p_0 \Exp(Y \mid X=0)}{p_1}.
\]
Hence bounds on the conditional mean can be obtained from bounds on the unconditional mean $\Exp(Y_0)$. Let 
\[
	\underline{E}_0^c(w) = \int_0^1 \underline{Q}_{Y_0}^c(\tau \mid w) \; d\tau 
	\qquad \text{and} \qquad
	\overline{E}_0^c(w) = \int_0^1 \overline{Q}_{Y_0}^c(\tau \mid w) \; d\tau
\]
denote bounds on $\Exp(Y_0 \mid W=w)$. Averaging these over the marginal distribution of $W$ yields bounds on $\Exp(Y_0)$, denoted by
\[
	\underline{E}_0^c = \Exp \big( \underline{E}_0^c(W) \big)
	\qquad \text{and} \qquad
	\overline{E}_0^c = \Exp \big( \underline{E}_0^c(W) \big).
\]

\begin{proposition}[Bounds on ATT]\label{prop:ATT bounds}
Suppose A\ref{assn:continuity}.\ref{A1_1}, A\ref{assn:continuity}.\ref{A1_2}, A\ref{assn:continuity}.\ref{A1_3}$^\prime$, and A\ref{assn:cdep}$^\prime$ hold. Suppose the joint distribution of $(Y,X,W)$ is known. Then $\text{ATT}$ lies in the set
\[
	\left[\Exp(Y \mid X=1) - \frac{\overline{E}_0^c - p_0\Exp(Y \mid X=0)}{p_1}, \, \Exp(Y \mid X=1) - \frac{\underline{E}_0^c - p_0\Exp(Y \mid X=0)}{p_1}\right]
\]
assuming these means exist (including possibly $\pm \infty$). Moreover, the interior of this set is sharp.
\end{proposition}

Bounds for the average effect of treatment on the untreated, $\text{ATU} = \Exp(Y_1 - Y_0 \mid X=0)$, can be obtained similarly.

% Here we use the fact that $\Exp(Y_0 \mid X=0)$ is point identified, and we compute bounds on $\Exp(Y_1 \mid X=0)$ using unconditional bounds on $\Exp(Y_1)$ derived earlier.

\subsection*{Bounds with binary outcomes}\label{sec:binaryOutcomes}

Here we drop the continuity assumption A\ref{assn:continuity}.\ref{A1_1} and instead consider binary potential outcomes $Y_x$. We replace the support assumption A\ref{assn:continuity}.\ref{A1_2} by the following.

\begin{assump}{A\ref{assn:continuity}.\ref{A1_2}$^\prime$}
For all $x,x' \in \{0,1\}$ and $w \in \supp(W)$, $\supp(Y_x \mid X=x', W=w) = \{ 0,1 \}$.
\end{assump}

This assumption is equivalent to $\Prob(Y_x = 1 \mid X=x', W=w) \in (0,1)$ for all $x,x' \in \{ 0,1 \}$ and $w \in \supp(W)$. Let $p_{1 \mid x,w} = \Prob(Y=1 \mid X=x,W=w)$. Define
\[
	\overline{P}^c_x(1 \mid w)
	= \min\left\{\frac{p_{1 \mid x,w}p_{x \mid w}}{p_{x \mid w} - c}\indicator(p_{x \mid w}>c)+\indicator(p_{x \mid w}\leq c), \, p_{1 \mid x,w}p_{x \mid w} + (1-p_{x \mid w})\right\}
\]
and
\[
	\underline{P}^c_x(1 \mid w) = \frac{ p_{1 \mid x,w}p_{x \mid w} }{\min\{p_{x \mid w} + c,1\}}.
\]

\begin{proposition}\label{prop:binary_outcomes}
%Suppose $Y_x \in \{ 0, 1 \}$ is binary for each $x \in \{ 0, 1 \}$. 
Suppose A\ref{assn:continuity}.\ref{A1_2}$^\prime$, A\ref{assn:continuity}.\ref{A1_3}, and A\ref{assn:cdep} hold. Suppose the joint distribution of $(Y,X,W)$ is known. Let $x \in \{ 0, 1 \}$ and $w \in \supp(W)$. Then
\[
	\Prob(Y_x = 1 \mid W=w) \in \left[ \underline{P}_x^c(1 \mid w),\overline{P}_x^c(1 \mid w) \right].
\]
Moreover, the interior of this set is sharp.
\end{proposition}

Bounds for $\Prob(Y_x=0 \mid W=w)$ obtain immediately by taking complements. Averaging over the marginal distribution of $W$ yields
\[
	\Prob(Y_x=1) \in \left[ \Exp \left( \underline{P}_x^c(1 \mid W) \right), \, \Exp \left( \overline{P}_x^c(1 \mid W) \right) \right]
\]
with a sharp interior.

Bounds for average treatment effects $\Prob(Y_1=1) - \Prob(Y_0=1)$ can be obtained by combining the bounds for each separate probability $\Prob(Y_x=1)$, $x \in \{ 0,1 \}$, similarly to equation \eqref{eq:CQTE_bounds}.

\subsection*{Numerical illustration}\label{sec:numericalIllustration}

We conclude this section with a brief numerical illustration. For $x=0,1$ and $w=0,1$, suppose the density of $Y \mid X=x,W=w$ is
\[
	f_{Y \mid X,W}(y \mid x,w) = \frac{1}{\gamma_X x + \gamma_W w + \sigma} \phi_{[-4,4]} \left( \frac{y - (\pi_X x + \pi_W w)}{\gamma_X x + \gamma_W w + \sigma} \right)
\]
where $\phi_{[-4,4]}$ is the pdf for the truncated standard normal on $[-4,4]$. $X$ and $W$ are binary with
\[
	\Prob(X=1) = p_1,
	\quad
	\Prob(W=1) = q,
	\quad \text{and} \quad
	\Prob(X=1 \mid W=w) = p_{1 \mid w}
\]
for $w=0,1$. We let $(\pi_X,\pi_W) = (1,1)$, $(\gamma_X,\gamma_W) = (0.1,0.1)$, $p_1 = 0.5$, and $q = 0.5$ in all dgps. We specify the choice of $p_{1 \mid w}$ and $\sigma$ below.

Under the conditional independence assumption, this dgp implies that treatment effects are heterogeneous, with an average treatment effect of $\text{ATE} = \pi_X = 1$. To examine the sensitivity of this finding to partial failure of conditional independence, figure \ref{numericalIllustration_identifiedSets} shows identified sets for ATE under conditional $c$-dependence for $c$ from zero to one. First consider the solid lines, which are the same in both plots. These correspond to the dgp with $(p_{1 \mid 1},p_{1 \mid 0}) = (0.6, 0.4)$ and $\sigma = 0.965$. We see that ATE under conditional independence is positive, and that this conclusion is robust to deviations of up to about $c=0.26$ from independence, but not to larger deviations. 

\begin{figure}[t]%[ht]
\centering
\includegraphics[width=80mm]{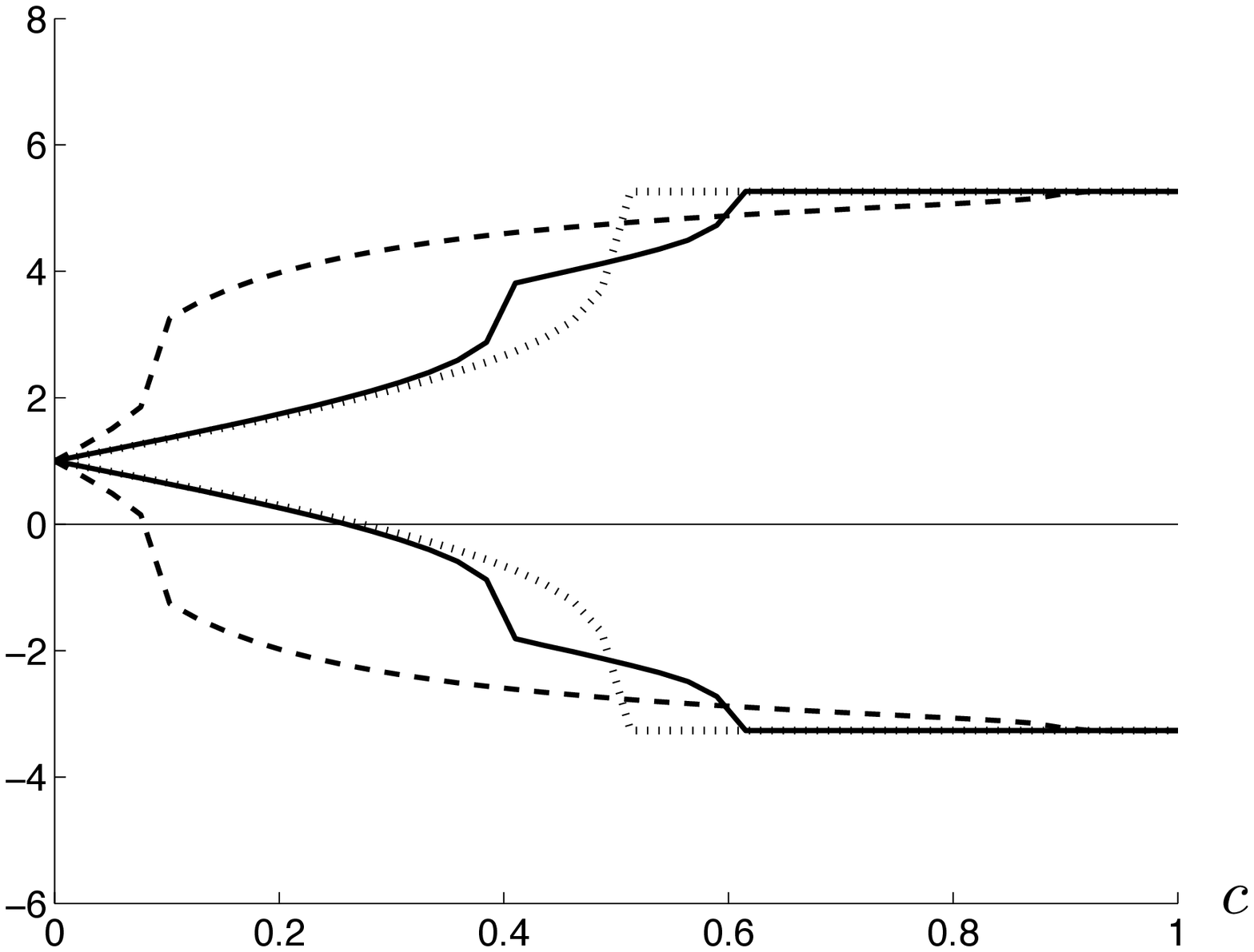}
%\hspace{4mm}
\includegraphics[width=80mm]{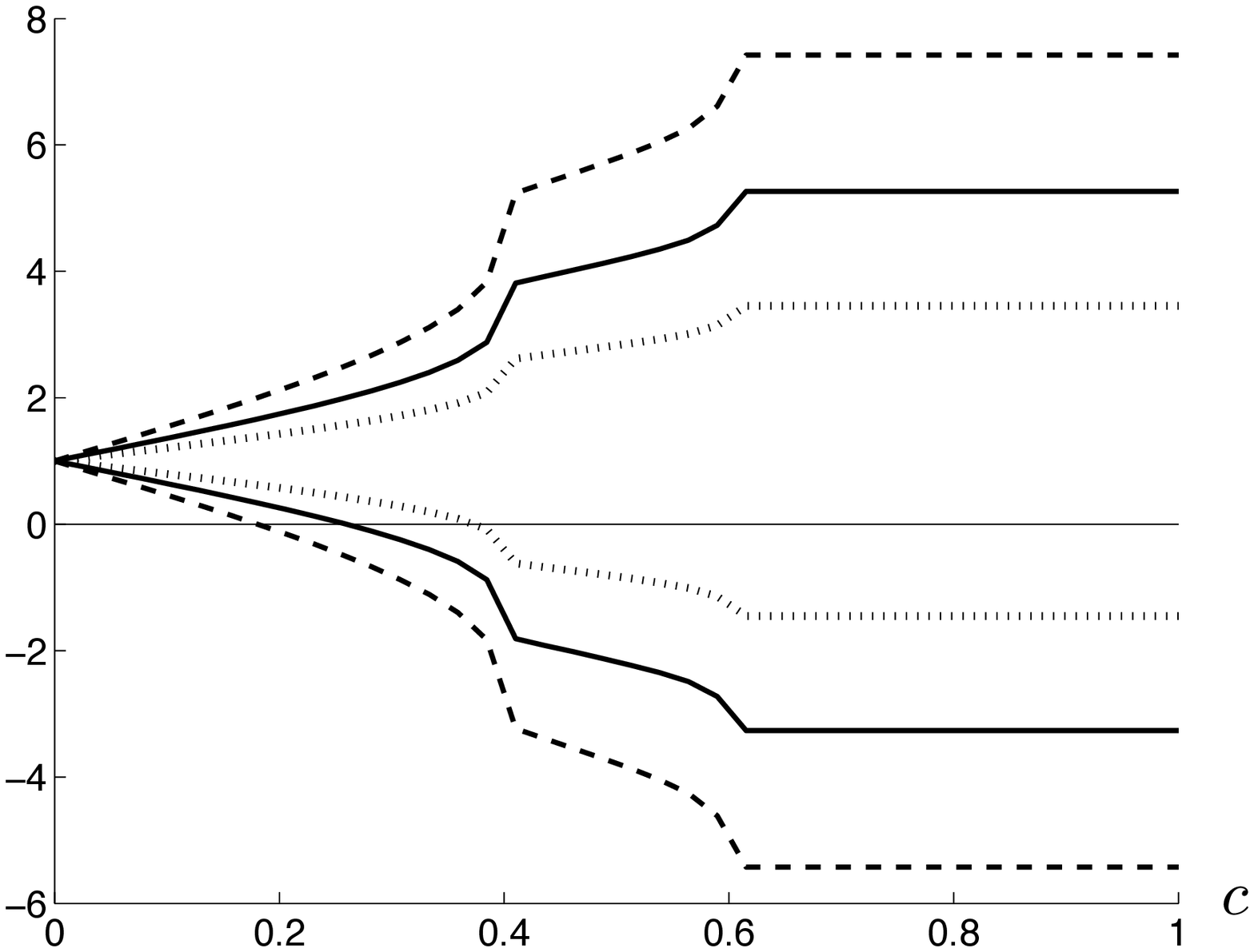} %\\
%\vspace{4mm}
%\includegraphics[width=70mm]{images/numericalIllustration_PI_ecta_plot1}
%\hspace{4mm}
%\includegraphics[width=70mm]{images/numericalIllustration_PI_ecta_plot4}
\caption{Identified sets for ATE, and how they depend on the dgp and the value of $c$. Left: For three dgps, corresponding to three values of the observed propensity score $p_{1 \mid 1}$ ($0.9$ dashed lines, $0.6$ solid lines, $0.5$ dotted lines). Right: For three dgps, corresponding to three values of $R^2$ (15\% dashed lines, 30\% solid lines, 60\% dotted lines).}
% (left two plots) and $\text{QTE}(0.5)$ (right two plots).}
\label{numericalIllustration_identifiedSets}
\end{figure}

Next we vary the dgp parameters to examine how our identified sets depend on features of the distribution of $(Y,X,W)$. \cite{Imbens2003} performed similar dgp comparisons for his method using empirical datasets. In the left plot we change the observed propensity score $p_{1 \mid w}$ while holding all other parameters fixed. Relative to the solid lines, if we increase the variation in the observed propensity score by setting $(p_{1 \mid 1}, p_{1 \mid 0}) = (0.9,0.1)$ then the bounds widen for most values of $c$, as shown by the dashed lines. In particular, the conclusion that ATE is positive now only holds for $c$'s less than about $0.085$. Conversely, if we eliminate the variation in the observed propensity score by setting $(p_{1 \mid 1}, p_{1 \mid 0}) = (0.5, 0.5)$ then the bounds shrink for most values of $c$, as shown by the dotted lines. The conclusion that ATE is positive now holds for slightly more values of $c$ than under the baseline dgp used for the solid lines.

Next consider the right plot. Here we change $R^2$ in the regression of $Y$ on $(1,X,W)$ while holding the observed propensity score fixed at $(p_{1 \mid 1},p_{1 \mid 0}) = (0.6,0.4)$. We vary the value of $R^2$ by varying $\sigma$. The solid lines have $R^2$ equal to 30\% (since $\sigma = 0.965$). Relative to these lines, if we decrease $R^2$ to 15\%, then the bounds widen for all values of $c$, as shown by the dashed lines. The conclusion that ATE is positive becomes less robust. Conversely, if we increase $R^2$ to 60\%, then the bounds shrink for all values of $c$, as shown by the dotted lines. The conclusion that ATE is positive becomes more robust.

The shape of the bounds depends on other features of the distribution of $(Y,X,W)$ as well. For example, if $\pi_X$ increases then all the identified sets shift upward. Hence, holding all else fixed, a larger ATE implies that the sign of ATE will be point identified under weaker independence assumptions. Similar analyses can also be done with other parameters of interest, like $\text{QTE}(\tau)$ for various values of $\tau$. Here we merely illustrate the kinds of objects empirical researchers can compute using the results we develop in this paper.

\section{Conclusion}\label{sec_conclusion}

In this paper we studied \emph{conditional $c$-dependence}, a nonparametric approach for weakening conditional independence assumptions. We used this concept to study identification of treatment effects when the conditional independence assumption partially fails, but no further data---like observations of an instrument---are available. We derived identified sets under conditional $c$-dependence for many parameters of interest, including average treatment effects and quantile treatment effects. These identified sets have simple, analytical characterizations. These analytical identified sets lend themselves to sample analog estimation and inference via the existing literature on inference under partial identification (see \citealt{CanayShaikh2017} for a survey). Our identification results can be used to analyze the sensitivity of one's results to the conditional independence assumption, without relying on auxiliary parametric assumptions.

Several questions remain. First, we focused on identification of $D$-parameters (\citealt{Manski2003}, page 11). Many other parameters, like the variance of potential outcomes or Gini coefficients, are not $D$-parameters. Nonetheless, the cdf and mean bounds we derived can be used as a direct input into theorem 2 of \cite{Stoye2010} to derive explicit, analytical bounds on these spread parameters. In future work it would be helpful to obtain precise expressions for these spread parameter bounds. Finally, while we have given several suggestions for how to interpret conditional $c$-dependence, there are likely other possibilities. For example, one could adapt Rosenbaum and Silber's \citeyearpar{RosenbaumSilber2009} `amplification' approach to our setting. Incorporating this or other ideas from the extensive literature on conditional treatment probabilities would be a helpful addition to our nonparametric sensitivity analysis.

\singlespacing
\bibliographystyle{econometrica}
\bibliography{BF_paper}

\appendix
\section{Appendix: Proofs}

\begin{proof}[Proof of proposition \ref{lemma:c-dep_KS_equivalence}]
By A\ref{assn:continuity}.\ref{A1_1},
\begin{align*}
	\Prob(X=1 \mid Y_x = y_x,W=w)
		&= \Prob(X=1 \mid F_{Y_x \mid W}(Y_x \mid W) = F_{Y_x \mid W}(y_x \mid W), W=w) \\
		&= \Prob(X=1 \mid R_x = r_x, W=w)
\end{align*}
where $r_x \equiv F_{Y_x \mid W}(y_x \mid w)$. Thus equation \eqref{eq:c-indep1} is equivalent to equation (\ref{eq:c-indep1}$^\prime$). It now suffices to show that equations (\ref{eq:c-indep1}$^\prime$) and \eqref{eq:jointMinusMarginalProductEquivalenceLemma} are equivalent. We have
\begin{align*}
	&| f_{X,R_x \mid W}(x',r \mid w) - p_{x' \mid w} f_{R_x \mid W}(r \mid w) | \\
		&\quad = | \Prob(X=x' \mid R_x=r,W=w) f_{R_x \mid W}(r \mid w) -  p_{x' \mid w}f_{R_x \mid W}(r \mid w) | \\
		&\quad = | \Prob(X=x' \mid R_x=r,W=w) -  p_{x' \mid w} | \cdot f_{R_x \mid W}(r \mid w)  \\
		&\quad = | \Prob(X=x' \mid R_x=r,W=w) -  p_{x' \mid w} |
\end{align*}
where the third equality follows since $R_x \mid W=w$ is uniformly distributed on $[0,1]$. Hence
\begin{multline*}
	\sup_{r \in [0,1]} | \Prob(X = x' \mid R_x = r,W=w) - \Prob(X=x' \mid W=w) | \\
	=
	\sup_{r \in [0,1]} | f_{X,R_x \mid W}(x',r \mid w) - p_{x' \mid w} f_{R_x \mid W}(r \mid w) |.
\end{multline*}
This holds for any $x'\in\{0,1\}$, which completes the proof.
\end{proof}

%%%%%%%%%%%%%%%%%%%%%%%%%%%%%%%%%%%%%%%%%%%%%%%%

The following lemma shows how to write conditional cdfs as integrals of conditional probabilities. We frequently use this result below.

\begin{lemma}\label{lemma:cdfAsIntegralOfPropensityScore}
Let $U$ be a continuous random variable. Let $X$ be a random variable with $p_x = \Prob(X=x) > 0$. Then
\[
	F_{U \mid X}(u \mid x) = \int_{-\infty}^u \frac{\Prob(X=x \mid U=v)}{p_x} \; dF_U(v).
\]
\end{lemma}

\begin{proof}[Proof of lemma \ref{lemma:cdfAsIntegralOfPropensityScore}]
We have
\begin{align*}
	p_x F_{U \mid X}(u \mid x)
		&= \Prob(X=x) \Prob(U \leq u \mid X=x) \\
		&= \Prob(U \leq u, X=x) \\
		&= \Exp[ \indicator(U \leq u) \indicator(X=x) ] \\
		&= \Exp( \Exp( \indicator(U \leq u) \indicator(X=x) \mid U) ) \\
		&= \Exp( \indicator(U \leq u) \Exp( \indicator(X=x) \mid U) ) \\
		&= \int_{-\infty}^\infty \indicator(v \leq u) \Exp( \indicator(X=x) \mid U=v) \; dF_U(v) \\
		&= \int_{-\infty}^u \Prob(X=x \mid U=v) \; dF_U(v).
\end{align*}
Now divide both sides by $p_x$.
\end{proof}

%%%%%%%%%%%%%%%%%%%%%%%%%%%%%%%%%%%%%%%%%%%%%%%%

\begin{proof}[Proof of theorem \ref{lemma:c-dep_cdf_bounds}]
This proof has five parts: (1) Show that $\overline{F}_{U \mid X}^c(\cdot \mid x)$ is an upper bound. (2) Show that $\underline{F}_{U \mid X}^c(\cdot \mid x)$ is a lower bound. (3) Show that these bound functions are valid cdfs. (4) Show that these bounds are sharp, in the sense stated in the theorem. (5) Apply these results to obtain the pointwise sharp bounds.

\bigskip	
\noindent \textbf{Part 1.} We show that $F_{U \mid X}(u \mid x) \leq \overline{F}_{U \mid X}^c(u \mid x)$ for all $u \in \supp(U)$.
\bigskip

Let $u \in \supp(U)$ be arbitrary. First, note that
\begin{align*}
	F_{U \mid X}(u \mid x)
	&= \int_{-\infty}^u \frac{\P(X=x \mid U = v)}{p_x} \; dF_U(v) \\
	&\leq \int_{-\infty}^u \frac{p_x + c}{p_x} \; dF_U(v) \\
	&= \left(1 + \frac{c}{p_x}\right) F_U(u).
\end{align*}
The first line follows by lemma \ref{lemma:cdfAsIntegralOfPropensityScore}. The second line follows by $c$-dependence (assumption 2). Likewise,
\begin{align*}
	F_{U \mid X}(u \mid x)
	&= 1- \int_u^\infty \frac{\P(X=x \mid U = v)}{p_x} \; dF_U(v) \\
	&\leq 1 - \int_u^\infty \frac{p_x - c}{p_x} \; dF_U(v) \; dv\\
	&=  \left(1 -\frac{c}{p_x}\right) F_U(u) + \frac{c}{p_x}.
\end{align*}
Also, since
\[
	F_{U \mid X}(u \mid x) = \frac{F_U(u) - p_{1-x} F_{U \mid X}(u \mid 1-x)}{p_x}
\]
by the law of iterated expectations and $F_{U \mid X}(u \mid 1-x) \geq 0$, we have
\[
	F_{U \mid X}(u \mid x) \leq \frac{F_U(u)}{p_x}.
\]
Finally, since $F_{U \mid X}(\cdot \mid x)$ is a cdf, it satisfies $F_{U \mid X}(u \mid x) \leq 1$. Therefore, $F_{U \mid X}(u \mid x)$ is smaller than each of the four functions inside the minimum in the definition of $\overline{F}_{U \mid X}^c(u \mid x)$. Thus it is smaller than the minimum too, and hence $F_{U \mid X}(u \mid x) \leq \overline{F}_{U \mid X}^c(u \mid x)$.

%%%%%%%%%%%%%%%%%%%%%%%%%%%%%%%%%%%%%%%%%%%%%%%%
\bigskip
\noindent \textbf{Part 2.} We show that $F_{U \mid X}(u \mid x) \geq \underline{F}_{U \mid X}^c(u \mid x)$ for all $u \in \supp(U)$.
\bigskip

Let $u \in \supp(U)$ be arbitrary. By a similar argument as in part 1,
\begin{align*}
	F_{U \mid X}(u \mid x)
	&= \int_{-\infty}^u \frac{\P(X=x \mid U = v)}{p_x} \; dF_U(v) \\
	&\geq \int_{-\infty}^u \frac{p_x - c}{p_x} \; dF_U(v) \\
	&= \left(1 - \frac{c}{p_x} \right) F_U(u)		
\end{align*}
and
\begin{align*}
	F_{U \mid X}(u \mid x)
	&= 1- \int_u^\infty \frac{\P(X=x \mid U = v)}{p_x} \; dF_U(v) \\
	&\geq 1 - \int_u^\infty \frac{p_x + c}{p_x} \; dF_U(v) \\
	&= \left(1 +\frac{c}{p_x}\right) F_U(u) - \frac{c}{p_x}.
\end{align*}
Also, since
\[
	F_{U \mid X}(u \mid x) = \frac{F_U(u) - p_{1-x} F_{U \mid X}(u \mid 1-x)}{p_x},
\]
$p_{1-x} = 1 - p_x$, and $F_{U \mid X}(u \mid 1-x) \leq 1$, we have that
\[
	F_{U \mid X}(u \mid x) \geq \frac{F_U(u)-1}{p_x} +1.
\]
Finally, since $F_{U \mid X}(\cdot \mid x)$ is a cdf, it satisfies $F_{U \mid X}(u \mid x) \geq 0$. Therefore, $F_{U \mid X}(u \mid x)$ is greater than each of the four functions inside the maximum in the definition of $\underline{F}_{U \mid X}^c(u \mid x)$. Thus it is greater than the maximum too, and hence $F_{U \mid X}(u \mid x) \geq \underline{F}_{U \mid X}^c(u \mid x)$.

%%%%%%%%%%%%%%%%%%%%%%%%%%%%%%%%%%%%%%%%%%%%%%%%
\bigskip
\noindent \textbf{Part 3.} The functions $\underline{F}_{U \mid X}^c(\cdot \mid x)$ and $\overline{F}_{U \mid X}^c(\cdot \mid x)$ are cdfs on $\supp(U)$.
\bigskip

By definition, $\underline{F}_{U \mid X}^c(\cdot \mid x)$ and $\overline{F}_{U \mid X}^c(\cdot \mid x)$ are compositions of continuous functions (e.g., the function $F_U(u)$, by assumption \ref{thm1_A2}), and hence are also continuous. Also by definition, they approach zero as $u$ approaches $\inf \supp(U)$ and approach one as $u$ approaches $\sup \supp(U)$.

When $c \leq p_x$, the expressions
\[
	F_U(u) - \frac{c}{p_x}\min\{F_U(u),1-F_U(u)\}
	\qquad \text{and} \qquad
	F_U(u) + \frac{c}{p_x}\min\{F_U(u),1-F_U(u)\}
\]
are nondecreasing in $u$. All other arguments of $\underline{F}_{U \mid X}^c(\cdot \mid x)$ and $\overline{F}_{U \mid X}^c(\cdot \mid x)$ are also nondecreasing. Hence $\underline{F}_{U \mid X}^c(\cdot \mid x)$ and $\overline{F}_{U \mid X}^c(\cdot \mid x)$ are nondecreasing when $c \leq p_x$.

When $c > p_x$, we have
\begin{align*}
	\overline{F}_{U \mid X}^c(u \mid x)
	&= \min \left\{F_U(u) + \frac{c}{p_x}\min\{F_U(u),1-F_U(u)\},\frac{F_U(u)}{p_x},1\right\}\\
	&= \min\left\{\left(1 + \frac{c}{p_x}\right)F_U(u), \frac{c}{p_x} + \left(1 - \frac{c}{p_x}\right)F_U(u),\frac{F_U(u)}{p_x},1\right\}\\
	&= \min\left\{\left(1 + \frac{c}{p_x}\right)F_U(u),\frac{F_U(u)}{p_x}, 1\right\}
\end{align*}
since
\begin{align*}
	\frac{c}{p_x} + \left(1 - \frac{c}{p_x}\right)F_U(u)
		&= \frac{c}{p_x} [1 - F_U(u)] + 1 \cdot F_U(u) \\
		&\geq 1
\end{align*}
for all $u \in \supp(U)$. The last line follows since $F_U(u) \in [0,1]$ and $c > p_x$ and therefore this term is a linear combination of one and something greater than one. Each of the three terms remaining in the expression for $\overline{F}_{U \mid X}^c(\cdot \mid x)$ is nondecreasing in $u$ and hence $\overline{F}_{U \mid X}^c(\cdot \mid x)$ is nondecreasing in $u$ when $c > p_x$.

Likewise, when $c > p_x$,
\begin{align*}
	\underline{F}_{U \mid X}^c(u \mid x)
	&= \max\left\{\left(1 - \frac{c}{p_x}\right) F_U(u) ,\left(1 + \frac{c}{p_x}\right)F_U(u) - \frac{c}{p_x} ,\frac{F_U(u)-1}{p_x} +1,0\right\}\\
	&= \max\left\{\left(1 + \frac{c}{p_x}\right)F_U(u) - \frac{c}{p_x} ,\frac{F_U(u)-1}{p_x} +1,0\right\}
\end{align*}
since $c > p_x$ implies
\[
	\left(1 - \frac{c}{p_x}\right) F_U(u) \leq 0
\]
for all $u \in \supp(U)$. Therefore, $\underline{F}_{U \mid X}^c(\cdot \mid x)$ is also nondecreasing when $c > p_x$. Thus we have shown that, regardless of the value of $c$, $\underline{F}_{U \mid X}^c(\cdot \mid x)$ and $\overline{F}_{U \mid X}^c(\cdot \mid x)$ are nondecreasing.

Putting all of these results together, we have shown that $\underline{F}_{U \mid X}^c(\cdot \mid x)$ and $\overline{F}_{U \mid X}^c(\cdot \mid x)$ satisfy all the requirements to be valid cdfs.

%%%%%%%%%%%%%%%%%%%%%%%%%%%%%%%%%%%%%%%%%%%%%%%%
\bigskip
\noindent \textbf{Part 4.} Sharpness.
\bigskip

We prove sharpness in two steps. First we construct a joint distribution of $(U,X)$ consistent with assumptions 1--4 and which yields the lower bound $\underline{F}_{U \mid X}^c(\cdot \mid x)$. And likewise for the upper bound $\overline{F}_{U \mid X}^c(\cdot \mid x)$. This yields equation \eqref{eq:intermediate R cdf bounds} for $\epsilon = 0$ and $\epsilon = 1$. Second we use certain linear combinations of these two joint distributions to obtain the case for $\epsilon \in (0,1)$.

The marginal distributions of $U$ and $X$ are prespecified. Hence to construct the joint distribution of $(U,X)$ it suffices to define conditional distributions of $X \mid U$. Specifically, for $c>0$ and $u\in\supp(U)$, define the conditional probabilities
\[
	\underline{P}_x^c(u)
	=
	\begin{cases}
		\max\{p_x - c,0\} &\text{if }  u \leq \underline{u}_x^c\\
		\min\{p_x + c,1\} &\text{if } \underline{u}_x^c < u
	\end{cases}
	\qquad
	\text{and}
	\qquad
	\overline{P}_x^c(u)
	=
	\begin{cases}
		\min\{p_x + c,1\} &\text{if }  u \leq \overline{u}_x^c \\
		\max\{p_x - c,0\} &\text{if } \overline{u}_x^c < u
	\end{cases}
\]
where
\[
	\underline{u}_x^c = F_U^{-1}\left(\frac{\min\{c,1-p_x\}}{\min\{p_x + c,1\} - \max\{p_x - c,0\}}\right)
\]
and
\[
	\overline{u}_x^c = F_U^{-1} \left(1-\frac{\min\{c,1-p_x\}}{\min\{p_x + c,1\} - \max\{p_x - c,0\}}\right).
\]
Note that $F_U(\overline{u}_x^c) + F_U(\underline{u}_x^c) = 1$. For $c=0$, define $\overline{P}_x^c(u) = \underline{P}_x^c(u) = p_x$. The definition of $\underline{u}_x^c$ above is derived as the number such that the lower bound function $\underline{F}_{U \mid X}^c(\cdot \mid x)$ can be written as
\[
	\underline{F}_{U \mid X}^c(u \mid x) =
	\begin{cases}
		\max \left\{ \left( 1 - \frac{c}{p_x} \right) F_U(u), 0 \right\}
			&\text{if $u \leq \underline{u}_x^c$} \\
		\max \left\{ \left( 1 + \frac{c}{p_x} \right) F_U(u) - \frac{c}{p_x}, \frac{F_U(u)-1}{p_x} + 1 \right\}
			&\text{if $u > \underline{u}_x^c$.}
	\end{cases}
\]
An example of this two-case form is shown in figure \ref{cIndepBoundsOnCdfs}. There we also see an example of the kink point $\underline{u}_x^c$. A similar result holds for the upper bound function, using the number $\overline{u}_x^c$.

The conditional probabilities $\underline{P}_x^c(u)$ and $\overline{P}_x^c(u)$ satisfy $c$-dependence by construction. Moreover, they are consistent with the marginal distribution of $X$, $\P(X=x) = p_x$, since
\begin{align*}
	\int_{\supp(U)} \underline{P}_x^c(u) \; dF_U(u)
	&= \max\{p_x - c,0\} \Prob(U \leq \underline{u}_x^c) + \min\{p_x+c,1\} \Prob(U > \underline{u}_x^c)\\
	&= \max\{p_x - c,0\} F_U \left( F_U^{-1} \left( \frac{\min\{c,1-p_x\}}{\min\{p_x+c,1\} - \max\{p_x-c,0\}} \right) \right) \\
	&\qquad + \min\{p_x+c,1\}\left[ 1- F_U \left( F_U^{-1} \left( \frac{\min\{c,1-p_x\}}{\min\{p_x+c,1\} - \max\{p_x-c,0\}} \right) \right) \right] \\
	&= \frac{(\max \{ p_x- c,0 \}  - \min \{ p_x+c,1 \} )\min \{ c,1-p_x \}}{\min \{ p_x+c,1 \} - \max \{ p_x-c,0 \} } + \min \{ p_x+c,1 \} \\
	&= - \min \{ c,1-p_x \} + \min \{ p_x+c,1 \} \\
	&= - \min \{ c,1-p_x \} + \min \{ c,1-p_x \} + p_x \\
	&= p_x
\end{align*}
and, by a similar proof,
\[
	\int_{\supp(U)} \overline{P}_x^c(u) \; dF_U(u) = p_x.
\]
To see that these conditional probabilities yield our bound functions, first let $\P(X=x \mid U=u) = \underline{P}_x^c(u)$ for all $u$. Then, by lemma \ref{lemma:cdfAsIntegralOfPropensityScore},
\begin{align*}
	\P(U\leq u  \mid X=x)
	&= \int_{-\infty}^u \frac{\underline{P}_x^c(v)}{p_x} \; dF_U(v)\\
	&= \underline{F}_{U \mid X}^c(u \mid x).
\end{align*}
To see this, note that for $u \leq \underline{u}_x^c$,
\begin{align*}
	\int_{-\infty}^u \frac{ \underline{P}_x^c(v)}{p_x} \; dF_U(v)
	&= \frac{\max \{ p_x - c, 0 \}}{p_x} F_U(u) \\
	&= \max \left\{ \left(1-\frac{c}{p_x} \right) F_U(u), 0 \right\}
\end{align*}
while for $u > \underline{u}_x^c$,
\begin{align*}
	\int_{-\infty}^u \frac{ \underline{P}_x^c(v)}{p_x} \; dF_U(v)
	&= \int_{-\infty}^{\infty} \frac{ \underline{P}_x^c(v)}{p_x} \; dF_U(v)  - \int_u^\infty \frac{ \underline{P}_x^c(v)}{p_x} \; dF_U(v)\\
	&= 1 - \min \left\{1 + \frac{c}{p_x},\frac{1}{p_x} \right\} (1-F_U(u)) \\
	&= \max \left\{ 1 - \left( 1 +\frac{c}{p_x} \right) (1-F_U(u)),1 - \frac{1-F_U(u)}{p_x} \right\}\\
	&= \max \left\{ \left( 1 + \frac{c}{p_x} \right) F_U(u) - \frac{c}{p_x}, \frac{F_U(u)- 1}{p_x}+1 \right\}.
\end{align*}
These two final expressions correspond to this lower bound. Similarly, letting $\P(X=x \mid U=u) = \overline{P}_x^c(u)$ for all $u$ yields $\P(U\leq u \mid X=x) = \overline{F}_{U \mid X}^c(u \mid x)$.

Thus we have shown that the bound functions are attainable. That is, equation \eqref{eq:intermediate R cdf bounds} holds with $\epsilon = 0$ or 1. Next consider $\epsilon \in (0,1)$. For this $\epsilon$, we specify the distribution of $X \mid U$ by the conditional probability $\epsilon \underline{P}_x^c(u) + (1-\epsilon) \overline{P}_x^c(u)$. This is a valid conditional probability since it is a linear combination of two terms which are between 0 and 1. Similarly, these two terms are between $p_x-c$ and $p_x+c$ and hence the linear combination is in $[p_x-c,p_x+c]$. Therefore this distribution satisfies $c$-dependence. By linearity of the integral and our results above, it yields
\[
	\int_{\supp(U)} [\epsilon \underline{P}_x^c(u) + (1-\epsilon) \overline{P}_x^c(u) ] \; dF_U(u) = p_x
\]
and hence is consistent with the marginal distribution of $X$. Again by linearity of the integral and our results above,
\[
	\Prob(U \leq u \mid X=x) = \epsilon \underline{F}_{U \mid X}^c(u \mid x) + (1-\epsilon)\overline{F}_{U \mid X}^c(u \mid x)
\]
as needed for equation \eqref{eq:intermediate R cdf bounds}.

\bigskip
\noindent \textbf{Part 5.} Pointwise sharp bounds.
\bigskip

We conclude by noting that the evaluation functional is monotonic (in the sense of first order stochastic dominance), which yields the pointwise bounds on $F_{U \mid X}(u \mid x)$. These are sharp by continuity of this functional, by continuity of the equation \eqref{eq:intermediate R cdf bounds} cdfs in $\epsilon$, and by varying $\epsilon$ over $[0,1]$.
\end{proof}

%%%%%%%%%%%%%%%%%%%%%%%%%%%%%%%%%%%%%%%%%%%%%%%%

\begin{proof}[Proof of proposition \ref{prop:condcdf_id}]
First we link the observed data to the unobserved parameters of interest:
\begin{align}
	F_{Y \mid X,W}(y \mid x,w)
		&= \P(Y \leq y  \mid X=x,W=w) \notag \\
		&= \P(Y_x \leq y  \mid X=x,W=w) \notag \\
		&= \P \big(F_{Y_x \mid W}(Y_x \mid W) \leq F_{Y_x \mid W}(y \mid w) \mid X=x,W=w \big)\notag \\
		&= F_{R_x \mid X,W}(F_{Y_x \mid W}(y \mid w) \mid x,w).\label{eq:condcdf_id prop eq1}
\end{align}
The second line follows by definition (equation \ref{eq:potential outcomes}). The third since $F_{Y_x \mid W}(\cdot \mid w)$ is strictly increasing (by A\ref{assn:continuity}.\ref{A1_1}). The fourth line by definition of $R_x$.

The left hand side of equation \eqref{eq:condcdf_id prop eq1} is known, while the argument of the right hand side is our parameter of interest. The main idea of this proof is that theorem \ref{lemma:c-dep_cdf_bounds} yields bounds on $F_{R_x \mid X,W}$, which we then invert to obtain bounds on $F_{Y_x \mid W}$. Showing this---part (1) below---is straightforward.
Several technical difficulties in proving sharpness arise, however, from the inversion step. These issues account for parts (2)--(6) below, as summarized next: (2) Define the functions $F_{Y_x \mid W}^c(\cdot \mid w; \epsilon,\eta)$. (3) Show that these functions are valid cdfs. (4) Show that, for any $\epsilon \in [0,1]$, these functions can be jointly attained. (5) Show that as $\eta$ converges to zero, $F_{Y_x \mid W}(\cdot \mid w; 0,\eta)$ approximates the lower bound from above. And likewise $F_{Y_x \mid W}(\cdot \mid w; 1,\eta)$ approximates the upper bound from below. (6) Show that $F_{Y_x \mid W}(y \mid w; \epsilon,\eta)$ is continuous when viewed as a function of $\epsilon$.

Finally, in part (7), we apply these results to obtain the pointwise bounds with sharp interior.

\bigskip
\noindent \textbf{Part 1}. First we show $F_{Y_x \mid W}(\cdot \mid w) \in \mathcal{F}_{Y_x \mid w}^c$.
\bigskip

By
\[
	\P(X=x \mid Y_x = y, W=w)
	=
	\P \Big( X=x \mid  F_{Y_x \mid W}(Y_x \mid W)=F_{Y_x \mid W}(y \mid w),W=w \Big)
\]
and conditional $c$-dependence (A\ref{assn:cdep}), we have
\[
  	\sup_{r \in [0,1]}|\P(X=x \mid R_x =r,W=w) - \P(X=x \mid W=w)| \leq c.
\]
Conditioning on $W=w$, apply theorem \ref{lemma:c-dep_cdf_bounds} to obtain bounds $\underline{F}_{R_x \mid X,W}^c$ and $\overline{F}_{R_x \mid X,W}^c$ for the distribution of $F_{R_x \mid X,W}$. Recall that $F_{R_x \mid W}(r \mid w) = r$ since $R_x \mid W=w \sim \text{Unif}[0,1]$.

For $\tau \in [0,1]$, let
\begin{align}
	\overline{Q}_x^c(\tau \mid w)
	&= \sup\{u\in[0,1] : \underline{F}_{R_x \mid X,W}^c(u \mid x,w) \leq \tau\}\notag\\
	&= \min\left\{\tau\frac{p_{x \mid w}}{p_{x \mid w}-c}\indicator(p_{x \mid w} > c) +\indicator(p_{x \mid w} \leq c),\frac{p_{x \mid w}\tau + c}{p_{x \mid w} + c},p_{x \mid w}\tau + (1-p_{x \mid w})\right\} \label{eq:Rquantile_lower}
\end{align}
denote the \emph{right}-inverse of $\underline{F}_{R_x \mid X,W}^c(\cdot \mid x,w)$. Similarly, let
\begin{align}
	\underline{Q}_x^c(\tau \mid w)
	&= \inf\{u\in[0,1] : \overline{F}_{R_x \mid X,W}^c(u \mid x,w) \geq \tau\}\notag\\
	&= \max\left\{\tau\frac{p_{x \mid w}}{p_{x \mid w}+c},\frac{p_{x \mid w}\tau - c}{p_{x \mid w} - c}\indicator(p_{x \mid w} > c),p_{x \mid w}\tau\right\} \label{eq:Rquantile_upper}
\end{align}
denote the \emph{left}-inverse of $\overline{F}_{R_x \mid X,W}^c(\cdot \mid x,w)$, respectively.

\bigskip

The bounds hold trivially, by definition, for $y \notin [\underline{y}_x(w),\overline{y}_x(w))$. Suppose $y\in [\underline{y}_x(w),\overline{y}_x(w))$.

\bigskip

\textbf{Lower bound:} By equation \eqref{eq:condcdf_id prop eq1} and since $\overline{F}_{R_x \mid X,W}^c$ is an upper bound,
\[
	F_{Y \mid X,W}(y \mid x,w)
	\leq \overline{F}_{R_x \mid X,W}^c(F_{Y_x \mid W}(y \mid w) \mid x,w).
\]
Hence
\begin{align*}
	\underline{F}^c_{Y_x \mid W}(y \mid w)
		&= \underline{Q}_x^c(F_{Y \mid X,W}(y \mid x,w) \mid w) \\
		&\leq \underline{Q}_x^c(\overline{F}_{R_x \mid X,W}^c(F_{Y_x \mid W}(y \mid w) \mid x,w) \mid w) \\
		&\leq F_{Y_x \mid W}(y \mid w).
\end{align*}
The first line follows by evaluating equation \eqref{eq:Rquantile_upper} at $\tau = F_{Y \mid X,W}(y \mid x,w)$, which yields equation \eqref{eq:cdf lowerbound}. The third line follows by \cite{Vaart2000} lemma 21.1 part (iv), for all $y\in [\underline{y}_x(w),\overline{y}_x(w))$.

\bigskip

\textbf{Upper bound:} We similarly have
\begin{align*}
	\overline{F}^c_{Y_x \mid W}(y \mid w)
		&= \overline{Q}_x^c(F_{Y \mid X,W}(y \mid x,w) \mid w) \\
		&\geq \overline{Q}_x^c(\underline{F}_{R_x \mid X,W}^c(F_{Y_x \mid W}(y \mid w) \mid x,w) \mid w) \\
		&\geq F_{Y_x \mid W}(y \mid w).
\end{align*}
The first line follows by evaluating equation \eqref{eq:Rquantile_lower} at $\tau = F_{Y \mid X,W}(y \mid x,w)$, which yields equation \eqref{eq:cdf upperbound}. The second line follows since, by equation \eqref{eq:condcdf_id prop eq1} and since $\underline{F}_{R_x \mid X,W}^c$ is a lower bound,
\begin{align*}
	F_{Y \mid X,W}(y \mid x,w) &\geq \underline{F}_{R_x \mid X,W}^c(F_{Y_x \mid W}(y \mid w) \mid x,w).
\end{align*}
The third and final line holds for all $y\in [\underline{y}_x(w),\overline{y}_x(w))$, and follows from
\begin{align*}
	\overline{Q}_x^c(\underline{F}_{R_x \mid X,W}^c(\tau \mid x,w) \mid w)
	&= \sup\{u\in[0,1]: \underline{F}_{R_x \mid X,W}^c(u \mid x,w) \leq  \underline{F}_{R_x \mid X,W}^c(\tau \mid x,w)\} \\
	&\geq \tau. % \tau is an element of this set, and therefore must be less than or equal to the supremum.
\end{align*}

Finally, note that without support independence (A\ref{assn:continuity}.\ref{A1_2}) our bound functions can be too tight and thus fail to be valid bounds. % This can most easily be seen by computing what we call the no assumptions bounds, which directly use $\supp(Y \mid X=x,W=w)$ to impute bounds on $\supp(Y_x \mid W=w)$.

%%%%%%%%%%%%%%%%%%%%%%%%%%%%%%%%%%%%%%%%%%%%%%%%
\bigskip
\noindent \textbf{Part 2.} Next we define the functions $F_{Y_x \mid W}^c(\cdot \mid w; \epsilon,\eta)$.
\bigskip

As in the proof of sharpness for theorem \ref{lemma:c-dep_cdf_bounds}, we will construct specific functions $\Prob(X=x \mid R_x=r, W=w)$. We then solve the equation
\[
	F_{Y \mid X,W}(y \mid x,w)
	= \int_0^{F_{Y_x \mid W}(y \mid w)} \frac{\P(X=x \mid R_x=r,W=w)}{p_{x \mid w}} \; dr
\]
(which holds by equation \eqref{eq:condcdf_id prop eq1} and lemma \ref{lemma:cdfAsIntegralOfPropensityScore}) for $F_{Y_x \mid W}$ to obtain our desired functions.

\bigskip

For $\eta \in (0,\min\{p_{x \mid w},1-p_{x \mid w}\})$ and $c>0$, let
\[
	\underline{P}_x^c(r,w;\eta) = \begin{cases}
	\max\{p_{x \mid w} - c,\eta\} &\text{if } 0 \leq r \leq \underline{u}_x^c(w;\eta)\\
	\min\{p_{x \mid w} + c,1-\eta\} &\text{if } \underline{u}_x^c(w;\eta) < r\leq 1
	\end{cases}
\]
and
\[
	\overline{P}_x^c(r,w;\eta) = \underline{P}_x^c(1-r,w;\eta)
\]
where
\[
	\underline{u}_x^c(w;\eta) = \frac{\min\{c,1-p_{x \mid w}-\eta\}}{\min\{p_{x \mid w} + c,1-\eta\} - \max\{p_{x \mid w} - c,\eta\}}.
\]
For $c=0$ let both of these functions be equal to $p_{x \mid w}$. For $\eta = 0$ these probabilities are simply those used in the proof of sharpness for theorem \ref{lemma:c-dep_cdf_bounds}, conditional on $W$. Using $\eta < \min\{p_{x \mid w},1-p_{x \mid w}\}$, it can be shown that the denominator used to define $\underline{u}_x^c(w;\eta)$ is always nonzero. This constraint on $\eta$ can also be used to show that $\underline{u}_x^c(w;\eta) \in [0,1]$. Also note that $\min\{p_{x \mid w},1-p_{x \mid w}\} > 0$ by A\ref{assn:continuity}.\ref{A1_3}, so that such positive $\eta$'s exist.

Define the function $G(\cdot;\epsilon,\eta): [0,1] \rightarrow [0,1]$ by
\begin{align*}
	G(d;\epsilon,\eta) &= \int_0^{d} \frac{\epsilon \underline{P}_x^c(r,w;\eta) + (1-\epsilon) \overline{P}_x^c(r,w;\eta)}{p_{x \mid w}} \; dr.
\end{align*}
$G$ also depends on $c$, $x$, and $w$ but we suppress this for simplicity. For $\eta \in (0,\min\{p_{x \mid w},1-p_{x \mid w}\})$ and for $\epsilon \in [0,1]$,
\[
	\epsilon \underline{P}_x^c(r,w;\eta) + (1-\epsilon) \overline{P}_x^c(r,w;\eta) > 0
\]
and hence $G(\cdot ; \epsilon,\eta)$ is strictly increasing; obtaining this property is a key reason why we use the $\eta$ variable. $G(\cdot ; \epsilon,\eta)$ is continuous. $G(0; \epsilon,\eta) = 0$. By derivations in part 4 below, $G(1; \epsilon,\eta) = 1$.  Thus $G(\cdot \mid \epsilon,\eta)$ is invertible with a continuous inverse $G^{-1}(\cdot ;\epsilon,\eta)$.

We thus define $F_{Y_x \mid W}^c(y \mid w;\epsilon,\eta)$ as the unique solution $d^*$ to
\[
	F_{Y \mid X,W}(y \mid x,w)
	= G(d^*;\epsilon,\eta).
	%= G(F_{Y_x \mid W}^c(y \mid w;\epsilon,\eta);\epsilon,\eta).
\]
That is,
\[
	F_{Y_x \mid W}^c(\cdot \mid w;\epsilon,\eta) = G^{-1}(F_{Y \mid X,W}(\cdot \mid x,w);\epsilon,\eta).
\]

\bigskip
\noindent \textbf{Part 3.} Show that these functions are valid cdfs.
\bigskip

$F_{Y \mid X,W}(\cdot \mid x,w)$ is continuous and strictly increasing by A\ref{assn:continuity}.\ref{A1_1}. Hence, for $\eta \in (0,\min\{p_{x \mid w},1-p_{x \mid w}\})$, $F_{Y_x \mid W}^c(\cdot \mid w;\epsilon,\eta)$ is the composition of two continuous and strictly increasing functions, and hence itself is continuous and strictly increasing. Since $G^{-1}(0; \epsilon,\eta) = 0$ and $G^{-1}(1; \epsilon,\eta) = 1$,  $F_{Y_x \mid W}^c(\cdot \mid w;\epsilon,\eta)$ equals zero when $y = \underline{y}_x(w)$ and equals one when $y = \overline{y}_x(w)$. Therefore it is a valid cdf.

%%%%%%%%%%%%%%%%%%%%%%%%%%%%%%%%%%%%%%%%%%%%%%%%

\bigskip
\noindent \textbf{Part 4.} Show that these functions can be jointly obtained.
\bigskip

To show this, we exhibit conditional probabilities $\P(X=x \mid R_1 = r,W=w)$ and $\P(X=x \mid R_0 = r,W=w)$ such that
\begin{enumerate}
\item They are consistent with the cdfs $F_{Y_1 \mid W}^c(y \mid w; \epsilon,\eta)$ and $F_{Y_0 \mid W}^c(y \mid w; 1-\epsilon,\eta)$ and with the observed cdf $F_{Y \mid X,W}$.

\item They are consistent with $p_{x \mid w}$.

\item They satisfy conditional $c$-dependence (A\ref{assn:cdep}) and A\ref{assn:continuity}.\ref{A1_1}.
\end{enumerate}
Specifically, for $\eta \in (0,\min\{p_{x \mid w},1-p_{x \mid w}\})$ we let
\[
	\P(X=x \mid R_1 = r,W=w)
	= \epsilon \underline{P}_x^c(r,w;\eta) + (1-\epsilon) \overline{P}_x^c(r,w;\eta)
\]
and
\[
	\P(X=x \mid R_0 = r,W=w)
	= (1-\epsilon) \underline{P}_x^c(r,w;\eta) + \epsilon \overline{P}_x^c(r,w;\eta).
\]
We now show that properties 1, 2, and 3 above hold for this choice:
\begin{enumerate}

\item This follows immediately by definition of the cdfs $F_{Y_x \mid W}^c(y \mid w; \epsilon,\eta)$ in part 2 above.

\item Recall that $R_x \mid W=w \sim \text{Unif}[0,1]$. Integrating against this marginal distribution yields
\begin{align*}
	&\hspace{4mm} \int_0^1 \underline{P}_x^c(r,w;\eta) \; dr \\
	&= \max\{p_{x \mid w}-c,\eta\} \underline{u}_x^c(w;\eta) + \min\{p_{x \mid w}+c,1-\eta\}(1-\underline{u}_x^c(w;\eta))\\
	&= \min\{p_{x \mid w}+c,1-\eta\} \\
	&\quad + \frac{\min\{c,1-p_{x \mid w}-\eta\}}{\min\{p_{x \mid w} + c,1-\eta\} - \max\{p_{x \mid w} - c,\eta\}}(\max\{p_{x \mid w} - c,\eta\} - \min\{p_{x \mid w}+c,1-\eta\})\\
	&= \min\{p_{x \mid w}+c,1-\eta\} - \min\{c,1-p_{x \mid w}-\eta\}\\
	&= p_{x \mid w} + \min\{c,1-p_{x \mid w}-\eta\} - \min\{c,1-p_{x \mid w}-\eta\} \\
	&= p_{x \mid w}
\end{align*}
similarly to derivations in part 4 of the proof of theorem \ref{lemma:c-dep_cdf_bounds}, and
\begin{align*}
	\int_0^1 \overline{P}_x^c(r,w;\eta) \; dr
	&= \int_0^1 \underline{P}_x^c(1-r,w;\eta) \; dr \\
	&= p_{x \mid w}.
\end{align*}
Hence
\[
	\int_0^1 [ \epsilon \underline{P}_x^c(r,w;\eta) + (1-\epsilon) \overline{P}_x^c(r,w;\eta)] \; dr = p_{x \mid w}.
\]

\item Conditional $c$-dependence holds by construction. A\ref{assn:continuity}.\ref{A1_1} holds as shown in part 3 above.
\end{enumerate}
Finally, note that we do not need to specify the joint distribution of $Y_1$ and $Y_0$ (or of $R_1$ and $R_0$) since the data only constraint the marginal distributions; any choice of copula is consistent with the data.

\bigskip
\noindent \textbf{Part 5}. Show that these functions monotonically approximate the bound functions.
\bigskip

We first derive explicit expressions for our approximating functions. Begin with the upper bound, which corresponds to $\epsilon = 1$. We have
\begin{align*}
	&\hspace{4mm} F_{Y \mid X,W}(y \mid x,w)p_{x \mid w} \\
	&= \int_0^{F_{Y_x \mid W}^c(y \mid w;1,\eta)} \underline{P}_x^c(r,w;\eta) \; dr \\
	&= \begin{cases}
	\max\{ p_{x \mid w} - c,\eta\} F_{Y_x \mid W}^c(y \mid w;1,\eta)
		&\text{if } F_{Y_x \mid W}^c(y \mid w;1,\eta) \leq \underline{u}_x^c(w;\eta) \\
	p_{x \mid w}- [1-F_{Y_x \mid W}^c(y \mid w;1,\eta)] \min \{ p_{x \mid w} + c,1-\eta \}
		&\text{if } F_{Y_x \mid W}^c(y \mid w;1,\eta) > \underline{u}_x^c(w;\eta)
	\end{cases}
\end{align*}
The first equality follows by definition of $F_{Y_x \mid W}(y \mid w;\epsilon,\eta)$. Solving for our approximating functions yields
\begin{align}\label{eq:condcdf_int_epsilon_2}
	&\hspace{4mm} F_{Y_x \mid W}^c(y \mid w;1,\eta) \notag \\
	&= \begin{cases}
	 \dfrac{F_{Y \mid X,W}(y \mid x,w) p_{x \mid w} }{\max \{p_{x \mid w} - c,\eta \}}
	 	&\text{if this expression is $\leq \underline{u}_x^c(w;\eta)$}
		%& \text{if } F\leq \underline{u}\\
	 \\
	 \dfrac{F_{Y \mid X,W}(y \mid x,w) p_{x \mid w} + \min \{c,1 - p_{x \mid w} - \eta \}}{\min \{ p_{x \mid w} + c,1-\eta \}}
	 	&\text{if this expression is $> \underline{u}_x^c(w;\eta)$}
		% &\text{if } F> \underline{u}
	 \end{cases} \\
	 &= \min\left\{\frac{p_{x \mid w}F_{Y \mid X,W}(y \mid x,w)}{\max\{p_{x \mid w} - c,\eta\}},\frac{F_{Y \mid X,W}(y \mid x,w)p_{x \mid w} + \min\{c,1-\eta - p_{x \mid w}\}}{\min\{p_{x \mid w}+c,1-\eta\}}\right\} \notag \\
	&= \min\left\{\frac{p_{x \mid w}F_{Y \mid X,W}(y \mid x,w)}{\max\{p_{x \mid w} - c,\eta\}},\frac{F_{Y \mid X,W}(y \mid x,w)p_{x \mid w} + c}{\min\{p_{x \mid w}+c,1-\eta\}},\frac{F_{Y \mid X,W}(y \mid x,w)p_{x \mid w} + (1-\eta) - p_{x \mid w}}{\min\{p_{x \mid w}+c,1-\eta\}}\right\}. \notag
\end{align}
The last line obtains by extracting the minimum in the numerator.
By similar calculations, for the lower bound ($\epsilon = 0$) we obtain
\begin{align}\label{eq:condcdf_int_epsilon_1}
	&F_{Y_x \mid W}^c(y \mid w;0,\eta) \\
	&= \max\left\{\frac{p_{x \mid w}F_{Y \mid X,W}(y \mid x,w)}{p_{x \mid w} + c},\frac{F_{Y \mid X,W}(y \mid x,w)p_{x \mid w} - \min\{c, p_{x \mid w}-\eta\}}{\max\{p_{x \mid w}-c,\eta\}},\frac{p_{x \mid w}F_{Y \mid X,W}(y \mid x,w)}{1-\eta}\right\}. \notag
\end{align}

Consider our approximation to the lower bound, $F_{Y_x \mid W}^c(y \mid w;0,\eta)$. The first piece of the maximum does not depend on $\eta$, and corresponds to the first piece in the limit function $\underline{F}_{Y_x \mid W}^c(y \mid w)$, equation \eqref{eq:cdf lowerbound}. For the third piece, as $\eta \searrow 0$,
\[
	\frac{p_{x \mid w}F_{Y \mid X,W}(y \mid x,w)}{1-\eta}
	\searrow
	p_{x \mid w} F_{Y \mid X,W}(y \mid x,w).
\]
This limit is the third piece of equation \eqref{eq:cdf lowerbound}. Finally consider the middle piece of our approximation function. If $p_{x \mid w} > c$ then for any $\eta \in (0,p_{x \mid w} - c)$, this middle piece exactly equals the middle piece of equation \eqref{eq:cdf lowerbound}. If $p_{x \mid w} \leq c$ then, as $\eta \searrow 0$,
\begin{align*}
	\frac{F_{Y \mid X,W}(y \mid x,w)p_{x \mid w} - \min\{c, p_{x \mid w}-\eta\}}{\max\{p_{x \mid w}-c,\eta\}}
	&= \frac{F_{Y \mid X,W}(y \mid x,w)p_{x \mid w} - p_{x \mid w} + \eta}{\eta} \\
	&= \frac{[ F_{Y \mid X,W}(y \mid x,w) - 1] p_{x \mid w} }{\eta} + 1 \\
	&\searrow -\infty.
\end{align*}
Hence this term disappears from the overall maximum. Thus we have shown that, for any $y \in \R$ and $w \in \supp(W)$,
\[
	F_{Y_x \mid W}^c(y \mid w;0,\eta) \searrow \underline{F}_{Y_x \mid W}^c(y \mid w)
\]
as $\eta \searrow 0$. A similar argument, based on our explicit expression for $F_{Y_x \mid W}^c(y \mid w;1,\eta)$, shows that
\[
	F_{Y_x \mid W}^c(y \mid w;1,\eta) \nearrow \overline{F}_{Y_x \mid W}^c(y \mid w)
\]
as $\eta \searrow 0$.

\bigskip
\noindent \textbf{Part 6}. Show that $F_{Y_x \mid W}^c(y \mid w; \cdot, \eta)$ is continuous on $[0,1]$.
%Show that, when viewed as functions of $\epsilon$, these functions are continuous.
%Show that, when viewed as functions of $\epsilon$ or $\eta$, these functions are continuous.
\bigskip

%For $\epsilon=0$ or 1, continuity of the function $F_{Y_x \mid W}^c(y \mid w; \epsilon,\cdot)$ on $(0,\min \{ p_{x \mid w}, 1 - p_{x \mid w} \})$ follows from our explicit expressions in part 5 above, which show that they are compositions of continuous functions.

Let $\{ \epsilon_n \} \subset [0,1]$ be a sequence converging to $\epsilon \in [0,1]$. Recall from part 2 that, by definition of  $F_{Y_x \mid W}^c(y \mid w; \epsilon, \eta)$,
\[
	F_{Y \mid X,W}(y \mid x,w) = G(F_{Y_x \mid W}^c(y \mid w;\epsilon_n,\eta); \epsilon_n, \eta)
\]
for all $n$. Taking limits as $n \rightarrow \infty$ on both sides yields
\[
	F_{Y \mid X,W}(y \mid x,w) = G \left( \lim_{n \rightarrow \infty} F_{Y_x \mid W}^c(y \mid w;\epsilon_n,\eta); \epsilon, \eta \right).
\]
Continuity of $G(\cdot ; \cdot, \eta)$ allows us to pass the limit inside. Finally, inverting $G(\cdot ; \epsilon,\eta)$ yields
\begin{align*}
	\lim_{n \rightarrow \infty} F_{Y_x \mid W}^c(y \mid w;\epsilon_n,\eta)
	&= G^{-1}(F_{Y \mid X,W}(y \mid x,w); \epsilon, \eta) \\
	&= F_{Y_x \mid W}^c(y \mid w;\epsilon,\eta)
\end{align*}
as desired.

\bigskip
\noindent \textbf{Part 7}. Apply these results to obtain the pointwise bounds.
\bigskip

Fix $y \in \R$. Let $e \in (\underline{F}_{Y_x \mid W}^c(y \mid w),\overline{F}_{Y_x \mid W}^c(y \mid w))$. By part 5, there exists an $\eta^* > 0$ such that $e \in [F_{Y_x \mid W}^c(y \mid w; 0,\eta^*), F_{Y_x \mid W}^c(y \mid w; 1,\eta^*)]$. By part 6, $F_{Y_x \mid W}^c(y \mid w; \cdot,\eta^*)$ is continuous on $[0,1]$. Hence by the intermediate value theorem, there exists an $\epsilon^* \in [0,1]$ such that $e = F_{Y_x \mid W}^c(y \mid w; \epsilon^*,\eta^*)$. Thus the value $e$ is attainable.

\end{proof}

%%%%%%%%%%%%%%%%%%%%%%%%%%%%%%%%%%%%%%%%%%%%%%%%

\begin{proof}[Proof of proposition \ref{prop:CQTE_bounds}]
Recall equation \eqref{eq:condcdf_id prop eq1},
\[
	F_{Y \mid X,W}(y \mid x,w)
	= F_{R_x \mid X,W}(F_{Y_x \mid W}(y \mid w) \mid x,w).
\]
By invertibility of $F_{Y \mid X,W}(\cdot \mid x,w)$,
\[
	y = Q_{Y \mid X,W}[ F_{R_x \mid X,W}(F_{Y_x \mid W}(y \mid w) \mid x,w) \mid x,w].
\]
By invertibility of $F_{Y_x \mid W}(\cdot \mid w)$, evaluate this equation at $y = Q_{Y_x \mid W}(\tau \mid w)$ to get
\[
	Q_{Y_x \mid W}(\tau \mid w)
	= Q_{Y \mid X,W}(F_{R_x \mid X,W}(\tau \mid x,w) \mid x,w).
\]
Let $\underline{F}_{R_x \mid X,W}^c(\tau \mid x,w)$ and $\overline{F}_{R_x \mid X,W}^c(\tau \mid x,w)$ denote the bounds on $F_{R_x \mid X,W}$ obtained by applying theorem \ref{lemma:c-dep_cdf_bounds} conditional on $W$ and using $R_x \mid W=w \sim \text{Unif}[0,1]$. This latter fact implies that these bounds are the same for $R_1$ and $R_0$. Hence we let $\underline{F}_{R \mid X,W}^c(\tau \mid x,w)$ and $\overline{F}_{R \mid X,W}^c(\tau \mid x,w)$ denote the common bounds. Substituting these bounds into the equation above yields the bounds \eqref{eq:quantile upperbound} and \eqref{eq:quantile lowerbound}. Taking the smallest and largest differences of these bounds yields \eqref{eq:CQTE_bounds}.

Here we see that proving proposition \ref{prop:CQTE_bounds} is simpler than proposition \ref{prop:condcdf_id} since we do not need to invert the $F_{R_x \mid X,W}$ bounds. To complete the proof, we prove sharpness using the same construction as in that proposition. Specifically, let $\eta \in (0,\min\{ p_{x \mid w}, 1 - p_{x \mid w} \})$. Then choose
\[
	\P(X=x \mid R_1 = r,W=w) = \epsilon \underline{P}_x^c(r,w;\eta) + (1-\epsilon) \overline{P}_x^c(r,w;\eta)
\]
and
\[
	\P(X=x \mid R_0 = r,W=w) = (1-\epsilon) \underline{P}_x^c(r,w;\eta) + \epsilon \overline{P}_x^c(r,w;\eta).
\]
These are attainable as in part 4 of the proof of proposition \ref{prop:condcdf_id}. By lemma \ref{lemma:cdfAsIntegralOfPropensityScore}, we can convert these to conditional probabilities to the cdfs
\[
	\widetilde{F}_{R_1 \mid X,W}^c(r \mid x,w;\epsilon,\eta) = G(r; \epsilon,\eta)
	\qquad \text{and} \qquad
	\widetilde{F}_{R_0 \mid X,W}^c(r \mid x,w;\epsilon,\eta) = G(r; 1-\epsilon,\eta).
\]
Using the properties of $G$, as in part 2 of the proof of proposition \ref{prop:condcdf_id}, we see that these are valid cdfs on $[0,1]$ which are strictly increasing in $r$ and continuous in $r$.

Setting $\eta = 0$ yields
\[
	\widetilde{F}_{R_x \mid X,W}^c(r \mid x,w; 1,0) = \underline{F}_{R \mid X,W}^c(r \mid x,w)
	\qquad \text{and} \qquad
	\widetilde{F}_{R_x \mid X,W}^c(r \mid x,w; 0,0) = \overline{F}_{R \mid X,W}^c(r \mid x,w),	
\]
which are not always strictly increasing. Hence, when substituted into $Q_{Y \mid X,W}(\cdot \mid x,w)$, the corresponding conditional quantile functions violate A\ref{assn:continuity}.\ref{A1_1}. As in part 5 of the proof of proposition \ref{prop:condcdf_id}, however, we can monotonically approximate these bound functions as $\eta \searrow 0$.

Substituting our constructed cdfs into our equation for $Q_{Y_x \mid W}(\tau \mid w)$ above and taking differences yields
\begin{align}\label{eq:CQTEsharpnessEq}
	Q_{Y_1 \mid W}(&\tau \mid w) - Q_{Y_0 \mid W}(\tau \mid w) \\
	&= Q_{Y \mid X,W}(\widetilde{F}_{R_1 \mid X,W}(\tau \mid 1,w;\epsilon,\eta) \mid 1,w) - Q_{Y \mid X,W}\left(\widetilde{F}_{R_0 \mid X,W}(\tau \mid 0,w;1-\epsilon,\eta) \mid 0,w\right). \notag
\end{align}
The final step now follows as in part 7 of the proof of proposition \ref{prop:condcdf_id}: Let $e \in (\underline{\text{CQTE}}^c(\tau \mid w), \overline{\text{CQTE}}^c(\tau \mid w) )$. By monotone approximation, there exists an $\eta^* > 0$ such that $e$ is above equation \eqref{eq:CQTEsharpnessEq} evaluated at $(\epsilon,\eta) = (1,\eta^*)$ and is below that equation evaluated at $(\epsilon,\eta) = (0,\eta^*)$. Next note that equation \eqref{eq:CQTEsharpnessEq} is continuous in $\epsilon$ since $Q_{Y \mid X,W}(\cdot \mid x,w)$ is continuous and by continuity of $\widetilde{F}_{R_x \mid X,W}(\tau \mid x,w;\cdot,\eta)$ on $[0,1]$. Thus, by the intermediate value theorem, there exists an $\epsilon^* \in [0,1]$ such that equation \eqref{eq:CQTEsharpnessEq} evaluated at $(\epsilon,\eta) = (\epsilon^*,\eta^*)$ yields $e$.

\end{proof}

%%%%%%%%%%%%%%%%%%%%%%%%%%%%%%%%%%%%%%%%%%%%%%%%

\begin{proof}[Proof of corollary \ref{corollary:CATE_ATE_bounds}]
\hfill
\begin{enumerate}
\item We obtain the CATE bounds by integrating the CQTE bounds in proposition \ref{prop:CQTE_bounds} over $\tau$. To show sharpness, we prove two results:
\begin{enumerate}
\item For any $\eta \in (0,\min\{ p_{x \mid w}, 1 - p_{x \mid w})$,
\[
	\int_0^1 Q_{Y \mid X,W}(\widetilde{F}_{R_x \mid X,W}^c(\tau \mid x,w;\epsilon,\eta) \mid x,w) \; d\tau
\]
is continuous in $\epsilon$ on $[0,1]$.

\item As $\eta \searrow 0$,
\[
	\int_0^1 Q_{Y \mid X,W}(\widetilde{F}_{R_x \mid X,W}^c(\tau \mid x,w;1,\eta) \mid x,w) \; d\tau
	\nearrow
	\int_0^1 Q_{Y \mid X,W}(\overline{F}_{R \mid X,W}^c(\tau \mid  x,w) \mid x,w) \; d\tau
\]
and
\[
	\int_0^1 Q_{Y \mid X,W}(\widetilde{F}_{R_x \mid X,W}^c(\tau \mid x,w;0,\eta) \mid x,w) \; d\tau
	\searrow
	\int_0^1 Q_{Y \mid X,W}(\underline{F}_{R \mid X,W}^c(\tau \mid  x,w) \mid x,w) \; d\tau.
\]
\end{enumerate}
We then use the same argument as in part 7 of the proof of proposition \ref{prop:condcdf_id}.

\bigskip

\noindent \textbf{Proof of (a)}.
Fix $\eta \in (0,\min\{ p_{x \mid w}, 1 - p_{x \mid w})$. First, notice that
\[
	\widetilde{F}_{R_x \mid X,W}^c(\tau \mid x,w;\epsilon,\eta)
	\in
	\left[ \widetilde{F}_{R_x \mid X,W}^c(\tau \mid x,w;1,\eta),
	\widetilde{F}_{R_x \mid X,W}^c(\tau \mid x,w;0,\eta) \right].
\]
Hence
\begin{multline*}
	Q_{Y \mid X,W}(\widetilde{F}_{R_x \mid X,W}^c(\tau \mid x,w;\epsilon,\eta) \mid x,w) \\
	\in
	\left[
		Q_{Y \mid X,W}(\widetilde{F}_{R_x \mid X,W}^c(\tau \mid x,w;1,\eta) \mid x,w),
		Q_{Y \mid X,W}(\widetilde{F}_{R_x \mid X,W}^c(\tau \mid x,w;0,\eta) \mid x,w)
	\right]
\end{multline*}
since $Q_{Y \mid X,W}(\cdot \mid x,w)$ is strictly increasing. Therefore,
\begin{align*}
	|Q_{Y \mid X,W}(\widetilde{F}_{R_x \mid X,W}^c(\tau \mid x,w;\epsilon,\eta) \mid x,w)|
		&\leq |Q_{Y \mid X,W}(\widetilde{F}_{R_x \mid X,W}^c(\tau \mid x,w;1,\eta) \mid x,w)| \\
		&\qquad + | Q_{Y \mid X,W}(\widetilde{F}_{R_x \mid X,W}^c(\tau \mid x,w;0,\eta) \mid x,w)|.
\end{align*}
Because of these bounds, it suffices to check that, for the endpoints $\epsilon=1$ and 0, the integral
\[
	\int_0^1 |Q_{Y \mid X,W}(\widetilde{F}_{R_x \mid X,W}^c(\tau \mid x,w;\epsilon,\eta) \mid x,w)| \; d\tau
\]
is finite. This boundedness then allows us to use the dominated convergence theorem to pass $\epsilon$ limits inside the integral. Continuity of the integral then follows since $Q_{Y \mid X,W}(\cdot \mid x,w)$ and $\widetilde{F}_{R_x \mid X,W}^c(\tau \mid x,w;\cdot,\eta)$ are continuous.

We finish this part by showing that those two integrals are finite. First consider the $\epsilon = 1$ case. Using our definitions (part 2 of the proof of proposition \ref{prop:condcdf_id}), computations similar to part 5 of the proof of proposition \ref{prop:condcdf_id} yield
\begin{align*}
	\widetilde{F}_{R_x \mid X,W}^c(\tau \mid x,w;1,\eta)
	&=
	\begin{cases}
		\max \left\{ \left(1 - \frac{c}{p_{x \mid w}} \right), \frac{\eta}{p_{x \mid w}} \right\} \tau
			&\text{ if } \tau \leq \underline{u}_x^c(w;\eta) \\
		1- \min \left\{ 1 + \frac{c}{p_{x \mid w}},\frac{1-\eta}{p_{x \mid w}} \right\} (1-\tau)
			&\text{ if } \tau \geq \underline{u}_x^c(w;\eta)
	\end{cases} \\
	&\equiv
	\begin{cases}
		A \tau
			&\text{ if } \tau \leq \underline{u}_x^c(w;\eta) \\
		1- B (1-\tau)
			&\text{ if } \tau \geq \underline{u}_x^c(w;\eta).
	\end{cases}
\end{align*}

%Let $Q^+_{Y|X,W}(\tau|x,w) = \max(Q_{Y|X,W}(\tau|x,w),0)$ and $Q^-_{Y|X,W}(\tau|x,w) = -\min(Q_{Y|X,W}(\tau|x,w),0)$.

Hence
\begin{align*}
	&\hspace{4mm} \int_0^1 |Q_{Y \mid X,W}( \widetilde{F}_{R_x \mid X,W}^c(\tau \mid x,w;1,\eta) \mid x,w)| \; d\tau \\
	&= \int_0^{\underline{u}_x^c(w;\eta)}
		|Q_{Y \mid X,W}\left( A \tau  \mid x,w \right)| d\tau
	+ \int_{\underline{u}_x^c(w;\eta)}^1
		|Q_{Y \mid X,W} \left( 1- B (1-\tau) \mid x,w \right)| d\tau \\
	&= \frac{1}{A} \int_0^{A \underline{u}_x^c(w;\eta) } |Q_{Y \mid X,W}(v \mid x,w)| \; dv
	 + \frac{1}{B} \int^1_{ 1- B (1-\underline{u}_x^c(w;\eta))} |Q_{Y \mid X,W}(v \mid x,w)| \; dv \\
	 &\leq \frac{1}{A} \int_0^1 |Q_{Y \mid X,W}(v \mid x,w)|\; dv
	 + \frac{1}{B} \int_0^1 |Q_{Y \mid X,W}(v \mid x,w)|; dv \\
	 &= (A^{-1} + B^{-1}) \cdot \Exp(|Y| \mid X=x,W=w) \\
	 &< \infty.
\end{align*}
The second equality follows by changing variables. The last line follows by assumption. Finally, the third line holds as follows:
\begin{itemize}
\item $\eta < p_{x \mid w}$ implies
\[
	A = \max \left\{ \left(1 - \frac{c}{p_{x \mid w}} \right), \frac{\eta}{p_{x \mid w}} \right\} \in (0,1]
\]
and hence $0 < A \underline{u}_x^c(w;\eta) \leq \underline{u}_x^c(w;\eta) \leq 1$.

\item Recalling that $\underline{u}_x^c(w;\eta) \in [0,1]$, we have
\[
	1 - B(1-\underline{u}_x^c(w;\eta)) = 1- \min \left\{ 1 + \frac{c}{p_{x \mid w}},\frac{1-\eta}{p_{x \mid w}} \right\} (1-\underline{u}_x^c(w;\eta)) \leq 1
\]
and
\[
%	1- \min \left\{ 1 + \frac{c}{p_{x \mid w}},\frac{1-\eta}{p_{x \mid w}} \right\} (1-\underline{u}_x^c(w;\eta))
	1 - B(1-\underline{u}_x^c(w;\eta))
	= \max \left\{ \frac{c}{p_{x \mid w}},\frac{p_{x \mid w} + \eta - 1}{p_{x \mid w}} \right\}
	+ \underline{u}_x^c(w;\eta) \min \left\{ \left(1 + \frac{c}{p_{x \mid w}} \right),\frac{1-\eta}{p_{x \mid w}} \right\}
	\geq 0.
\]
Finally, note that $B > 0$.
\end{itemize}
The proof for $\epsilon = 0$ is similar.

\bigskip

\noindent \textbf{Proof of (b)}. From the proof of (a),
\[
	\int_0^1 Q_{Y \mid X,W}(\widetilde{F}_{R_x \mid X,W}^c(\tau \mid x,w;1,\eta) \mid x,w) \; d\tau < \infty
\]
and
\[
	\int_0^1 Q_{Y \mid X,W}(\widetilde{F}_{R_x \mid X,W}^c(\tau \mid x,w;0,\eta) \mid x,w) \; d\tau > -\infty
\] % For instance, take the absolute value on the outside of the integral, then bring it on the inside, and then we have the results from part (a)
for any $\eta \in (0,\min \{ p_{x \mid w}, 1 - p_{x \mid w} \})$. The integrands converge pointwise monotonically to the limit functions $\overline{F}_{R \mid X,W}^c(\tau \mid x,w)$ and $\underline{F}_{R \mid X,W}^c(\tau \mid x,w)$, respectively, as $\eta \searrow 0$. The result now follows by the monotone convergence theorem (e.g., theorem 4.3.2 on page 131 of \citealt{Dudley2002}).

\item The ATE bounds follow by integrating the CATE bounds over the marginal distribution of $W$. Next we show sharpness. We consider the case $c > 0$, so that our bounds have nonempty interior. The case $c=0$ is trivial as this is the usual unconfoundedness result.

Let
\[
	C(w) = \Exp(Y \mid X=1,W=w) - \Exp(Y \mid X=0, W=w).
\]
$C(w)$ is in the identified set for $\text{CATE}(w)$. Moreover, $\Exp( | C(W) | ) < \infty$ by assumption. Define the functional $H$ on the set of functions $f$ satisfying $\Exp(| f(W) | ) < \infty$ by
\[
	H(f) = \int_{\supp(W)} f(w) \; dF_W(w).
\]
$H$ is continuous in the sup-norm, monotonic in the pointwise order on functions, and $\text{ATE} = H( \text{CATE}(\cdot) )$, assuming this mean exists. $H(C)$ is finite. Since $c > 0$, $H(C) \in (\underline{\text{ATE}}^c,\overline{\text{ATE}}^c )$.

Let $e \in (\underline{\text{ATE}}^c,\overline{\text{ATE}}^c )$. We want to find a function $f(\cdot)$ such that
\[
	f(w) \in \big( \underline{\text{CATE}}^c(w),\overline{\text{CATE}}^c(w) \big)
\]
and $H(f) = e$. First suppose $e - H(C) > 0$. Then we're looking for a function $f$ that is sufficiently above $C$, but does not violate the CATE bounds. Define
\[
	f_M(w) = \min \{ \overline{\text{CATE}}^c(w), C(w) + M \}
\]
for $M > 0$. Then $H(f_0) = H(C)$. Moreover, for each $w \in \supp(W)$, $f_M(w) \nearrow \overline{\text{CATE}}^c(w)$ as $M \nearrow \infty$. Thus $H(f_M) \nearrow H(\overline{\text{CATE}}^c(\cdot)) = \overline{\text{ATE}}^c$ by the monotone convergence theorem. Thus, since $e < \overline{\text{ATE}}^c$, there exists an $\overline{M}$ such that $H(f_{\overline{M}}) > e$. Finally, we note that $f_M(w)$ is continuous in $M$ and hence $H(f_M)$ is continuous in $M$. Thus by the intermediate value theorem there exists an $0 < M^* < \overline{M}$ such that $H(f_{M^*}) = e$. Thus $e$ is attainable. A similar argument applies if $e - H(C) \leq 0$.
% Case with equality---just choose $C(w)$!
\end{enumerate}
\end{proof}

%%%%%%%%%%%%%%%%%%%%%%%%%%%%%%%%%%%%%%%%%%%%%%%%

\begin{proof}[Proof of corollary \ref{lemma:margCDF_QTE_bounds}]
\hfill
\begin{enumerate}
\item We obtain bounds on $F_{Y_x}$ by integrating the bounds in proposition \ref{prop:condcdf_id} with respect to the marginal distribution of $W$. By the dominated convergence theorem,
\begin{align*}
	\lim_{\eta \searrow 0}\Exp[F_{Y_x \mid W}^c(y \mid W;0,\eta)]
	&= \Exp \left( \lim_{\eta \searrow 0} F_{Y_x \mid W}^c(y \mid W;0,\eta) \right) \\
	&= \underline{F}_{Y_x}(y).
\end{align*}
The second line follows by proposition \ref{prop:condcdf_id}. A similar argument applies for the upper bound. Also by the dominated convergence theorem, $\Exp [ F_{Y_x \mid W}^c(y \mid W; \epsilon, \eta) ]$ is continuous in $\epsilon \in [0,1]$, for all $\eta \in (0, \min \{ p_{x \mid w}, 1 - p_{x \mid w} \})$. The argument now proceeds as in part 7 of the proof of proposition \ref{prop:condcdf_id}.

\item The bounds on $Q_{Y_x}(\tau)$ are just the inverse of the bounds on $F_{Y_x}(\tau)$ from part 1. The cdf $\Exp[F_{Y_x \mid W}^c(y \mid W;0,\eta)]$ converges pointwise to $\underline{F}_{Y_x}(y)$ from above by part 1. Therefore, for any sequence $\eta_n \geq \eta_{n+1} \geq \ldots >0$,
\begin{align*}
	\inf \big\{y \in \R : \Exp[F_{Y_x \mid W}^c(y \mid W;0,\eta_n)] \geq \tau \big\} &\leq \inf \big\{ y \in \R : \Exp[F_{Y_x \mid W}^c(y \mid W;0,\eta_{n+1})] \geq \tau \big\} \leq \ldots
\end{align*}
This sequence thus converges monotonically to
\[
	\inf \{ y \in \R : \underline{F}_{Y_x}(y)\geq \tau\} = \overline{Q}_{Y_x}(\tau),
\]
% By monotone convergence (theorem 3.14 on page 55 of \citealt{Rudin1976}), this sequence converges to
which may be $+\infty$. A similar argument applies for the lower bound. Moreover, for any fixed $\eta \in (0, \min \{ p_{x \mid w}, 1 - p_{x \mid w} \})$, the inverse of $\Exp [ F_{Y_x \mid W}^c(\cdot \mid W; \epsilon, \eta) ]$ at $\tau$ is continuous in $\epsilon$ over $[0,1]$ (by a proof similar to part 6 of the proof of proposition \ref{prop:condcdf_id}). The argument now proceeds as in part 7 of the proof of proposition \ref{prop:condcdf_id}.

\item This result for $\text{QTE}(\tau)$ follows by combining the bounds on $Q_{Y_1}(\tau)$ and $Q_{Y_0}(\tau)$ from part 2, analogously to equation \eqref{eq:CQTE_bounds}, noting that the joint identified set is the product of the marginal identified sets.
\end{enumerate}
\end{proof}

%%%%%%%%%%%%%%%%%%%%%%%%%%%%%%%%%%%%%%%%%%%%%%%%

\begin{proof}[Proof of proposition \ref{prop:ATT bounds}]
By the law of total probability and the definition of ATT,
\begin{equation}\label{eq:ATT proof eq 1}
	\text{ATT}
	= \Exp(Y \mid X=1) - \frac{\Exp(Y_0) - p_0\Exp(Y \mid X=0)}{p_1}.
\end{equation}
Using assumptions A\ref{assn:continuity}.\ref{A1_1}, A\ref{assn:continuity}.\ref{A1_2}, A\ref{assn:continuity}.\ref{A1_3}$^\prime$, and A\ref{assn:cdep}$^\prime$, a proof similar to that of corollary \ref{corollary:CATE_ATE_bounds} shows that the set $(\underline{E}_0^c,\overline{E}_0^c)$ is sharp. Substituting these bounds in equation \eqref{eq:ATT proof eq 1} completes the proof. 
\end{proof}

\begin{proof}[Proof of proposition \ref{prop:binary_outcomes}]
First, we show that $[\underline{P}_x^c(1 \mid w),\overline{P}_x^c(1 \mid w)]$ are bounds for $\P(Y_x=1 \mid W=w)$. We then show sharpness of the interior.

When $c < p_{x \mid w}$,
\begin{align*}
	\P(Y_x=1 \mid W=w)
		&= \frac{\Prob(Y_x=1,W=w)}{\Prob(W=w)} \frac{\Prob(Y_x=1, X=x, W=w)}{\Prob(Y_x=1, X=x, W=w)} \frac{\Prob(X=x,W=w)}{\Prob(X=x,W=w)} \\
		&= \frac{p_{1 \mid x,w}p_{x \mid w}}{\P(X=x \mid Y_x=1,W=w)} \\
		&\leq \frac{p_{1 \mid x,w}p_{x \mid w}}{p_{x \mid w} - c}.
\end{align*}
The last line follows by conditional $c$-dependence. Also,
\begin{align*}
	\P(Y_x=1 \mid W=w)
		&= \P(Y=1 \mid X=x,W=w)p_{x \mid w} + \P(Y_x=1 \mid X=1-x,W=w)(1-p_{x \mid w}) \\
		&\leq p_{1 \mid x,w}p_{x \mid w} + 1 \cdot (1-p_{x \mid w}).
\end{align*}
Combining these two inequalities yields $\P(Y_x=1 \mid W=w) \leq \overline{P}_x^c(1 \mid w)$. % (When $c \geq p_{x \mid w}$ the first piece in the definition of $\overline{P}_x^c$ is always 1, which is the logical upper bound on the probability.)

Similarly,
\begin{align*}
	\P(Y_x=1 \mid W=w)
		&=  \frac{p_{1 \mid x,w}p_{x \mid w}}{\P(X=x \mid Y_x=1,W=w)} \\
		&\geq \frac{p_{1 \mid x,w}p_{x \mid w}}{\min\{p_{x \mid w}+c,1\}}
\end{align*}
by conditional $c$-dependence and by $\P(X=x \mid Y_x=1,W=w) \leq 1$.

To show sharpness of the interior, fix $p^*\in (\underline{P}_x^c(1 \mid w),\overline{P}_x^c(1 \mid w))$. We want to exhibit a joint distribution of $(Y_x,X,W)$ consistent with the data, our assumptions, and which yields this element $p^*$. Since the distribution of $(X,W)$ is observed, we only need to specify the distribution of $Y_x \mid X,W$. Since $Y_x$ and $X$ are binary, there are only two parameters to this distribution to specify, for each $w \in \supp(W)$. The first is $\Prob(Y_x=1 \mid X=x,W=w) = \Prob(Y=1 \mid X=x,W=w)$, which is point identified from the data. Hence the only unknown parameter is the value $\Prob(Y_x=1 \mid X=1-x,W=w)$, which must be chosen such that
\begin{enumerate}
\item $\Prob(Y_x=1 \mid W=w) = p^*$,

\item$\Prob(Y_x=1 \mid X=1-x,W=w) \in (0,1)$, and

\item conditional $c$-dependence is satisfied.
\end{enumerate}

\noindent \textbf{Proof of 1.} Choose
\[
	\Prob(Y_x=1 \mid X=1-x,W=w) = \frac{p^* - p_{1 \mid x,w}p_{x \mid w}}{1-p_{x \mid w}}.
\]
Then,
\begin{align*}
	\P(Y_x=1 \mid W=w)
		&= \P(Y=1 \mid X=x,W=w)p_{x \mid w} + \Prob(Y_x \mid X=1-x,W=w)(1-p_{x \mid w}) \\
		&= p^*.
\end{align*}

\bigskip
\noindent \textbf{Proof of 2.} We have
\begin{align*}
	\Prob(Y_x = 1 \mid X=1-x,W=w) &= \frac{p^* - p_{1 \mid x,w}p_{x \mid w}}{1-p_{x \mid w}} \\
	&> \frac{\underline{P}_x^c(1 \mid w) - p_{1 \mid x,w}p_{x \mid w}}{1-p_{x \mid w}}\\
	 &= \frac{p_{1 \mid x,w}p_{x \mid w}}{1-p_{x \mid w}}\left(\frac{1}{\min\{p_{x \mid w}+c,1\}} - 1\right) \\
	&\geq 0
\end{align*}
where the second line follows by our choice of $p^*$. Similarly,
\begin{align*}
	\frac{p^* - p_{1 \mid x,w}p_{x \mid w}}{1-p_{x \mid w}}
	&< \frac{\overline{P}_x^c(1 \mid w) - p_{1 \mid x,w}p_{x \mid w}}{1-p_{x \mid w}} \\
	&= \frac{\min\left\{\frac{p_{1 \mid x,w}p_{x \mid w}}{p_{x \mid w} - c}\indicator(p_{x \mid w}>c)+\indicator(p_{x \mid w}\leq c),p_{1 \mid x,w}p_{x \mid w} + (1-p_{x \mid w})\right\} - p_{1 \mid x,w}p_{x \mid w}}{1-p_{x \mid w}} \\
	&= \min\left\{ \frac{1}{1-p_{x \mid w}} \left[\frac{p_{1 \mid x,w}p_{x \mid w}}{p_{x \mid w} - c}\indicator(p_{x \mid w}>c)+\indicator(p_{x \mid w}\leq c) -p_{1 \mid x,w} p_{x \mid w}\right],1\right\} \\
	&\leq 1.
\end{align*}
Hence $\Prob(Y_x=1 \mid X=1-x,W=w)\in(0,1)$ is a valid probability which satisfies assumption A\ref{assn:continuity}.\ref{A1_2}$^\prime$.

\bigskip
\noindent \textbf{Proof of 3.} By Bayes' rule, we have
\[
	\P(X=1 \mid Y_x=1,W=w) = \frac{p_{1 \mid 1,w}p_{1 \mid w}}{\P(Y_x=1 \mid W=w)}.
\]
By the lower bound on $\P(Y_x=1 \mid W=w)$, we have
\begin{align*}
	\P(X=1 \mid Y_x=1,W=w)
		& \leq \frac{p_{1 \mid 1,w}p_{1 \mid w}}{\underline{P}_1^c(1 \mid w)} \\
		&= \min\{p_{1 \mid w}+c,1\} \\
		&\leq p_{1 \mid w} + c.
\end{align*}
By the upper bound on $\P(Y_x=1 \mid W=w)$, we have
\begin{align*}
	\Prob(X=1 \mid Y_x=1,W=w)
	& \geq \frac{p_{1 \mid 1,w}p_{1 \mid w}}{\overline{P}_1^c(1 \mid w)} \\
	&=\frac{p_{1 \mid 1,w}p_{1 \mid w}}{\min\left\{\frac{p_{1 \mid 1,w}p_{1 \mid w}}{p_{1 \mid w} - c}\indicator(p_{1 \mid w}>c)+\indicator(p_{1 \mid w}\leq c),p_{1 \mid 1,w}p_{1 \mid w} + (1-p_{1 \mid w})\right\}}.
\end{align*}
When $p_{1 \mid w}>c$, this gives
\begin{align*}
	\Prob(X=1 \mid Y_x=1,W=w)
	&\geq \max\left\{p_{1 \mid w}-c,\frac{p_{1 \mid 1,w}p_{1 \mid w}}{p_{1 \mid 1,w}p_{1 \mid w} + (1-p_{1 \mid w})}\right\}\\
	&\geq p_{1 \mid w} - c.
\end{align*}
If $p_{1 \mid w} \leq c$, then $0 \geq p_{1 \mid w} - c$ and hence $\P(X=1 \mid Y_x=1,W=w) \geq p_{1 \mid w} - c$ holds trivially.

Therefore $\P(X=1 \mid Y_x=1,W=w)  \in [p_{1 \mid w} - c,p_{1 \mid w}+c]$. A similar calculation yields the same result for $\P(X=1 \mid Y_x=0,W=w)$.
\end{proof}

\end{document}